\documentclass[11pt,a4paper]{article}
\usepackage{graphicx}
\usepackage{url}
\usepackage{amsfonts,amsmath,amssymb}
\usepackage{amsthm}
\usepackage{subfig}
\usepackage{hyperref}
\usepackage{fullpage}
\usepackage{xcolor}
\usepackage{mathtools}
\usepackage[mathlines]{lineno}
\usepackage{multirow}
\usepackage{authblk}
\usepackage{algorithm}
\usepackage{pstricks}
\usepackage{pst-node}
\let\doendproof\endproof
\renewcommand\endproof{~\hfill$\qed$\doendproof}

\usepackage{epsfig,xspace,algorithm,verbatim}
\usepackage[noend]{algpseudocode}

\newtheorem{theorem}{Theorem}%[section]
\newtheorem{proposition}[theorem]{Proposition}
\newtheorem{lemma}[theorem]{Lemma}
\newtheorem{corollary}[theorem]{Corollary}
\newtheorem{definition}[theorem]{Definition}

\newtheorem{observation}[theorem]{Observation}

\newcommand{\Kl}{{\rm Kernel }}
\newcommand{\perim}{\Pi}

\newcommand\blfootnote[1]{%
  \begingroup
  \renewcommand\thefootnote{}\footnote{#1}%
  \addtocounter{footnote}{-1}%
  \endgroup}

\title{Optimizing generalized kernels of polygons}
\date{}

\author[1]{Alejandra Martinez-Moraian\thanks{Email: \texttt{alejandra.martinezm@uah.es}}}
\author[2]{David Orden\thanks{Email: \texttt{david.orden@uah.es}}}
\author[3]{Leonidas Palios\thanks{Email: \texttt{palios@cs.uoi.gr}}}
\author[4]{Carlos Seara\thanks{Email: \texttt{carlos.seara@upc.edu}}}
\author[5]{Pawe{\l} \.{Z}yli\'{n}ski\thanks{Email: \texttt{zylinski@inf.ug.edu.pl}}}

\affil[1]{Departamento de F\'{i}sica y Matem\'{a}ticas, Universidad de Alcal\'{a}, Spain.}
\affil[2]{Departamento de F\'{i}sica y Matem\'{a}ticas, Universidad de Alcal\'{a}, Spain.}
\affil[3]{Dept.\ of Computer Science and Engineering, University of Ioannina, Greece.}
\affil[4]{Departament de Matem\`{a}tiques, Universitat Polit\`{e}cnica de Catalunya, Spain.}
\affil[5]{Institute of Informatics, University of Gda\'{n}sk, Poland.}

\begin{document}
  \maketitle
 \blfootnote{
 	\begin{minipage}[l]{0.3\textwidth}
 		\includegraphics[trim=10cm 6cm 10cm 5cm,clip,scale=0.12]{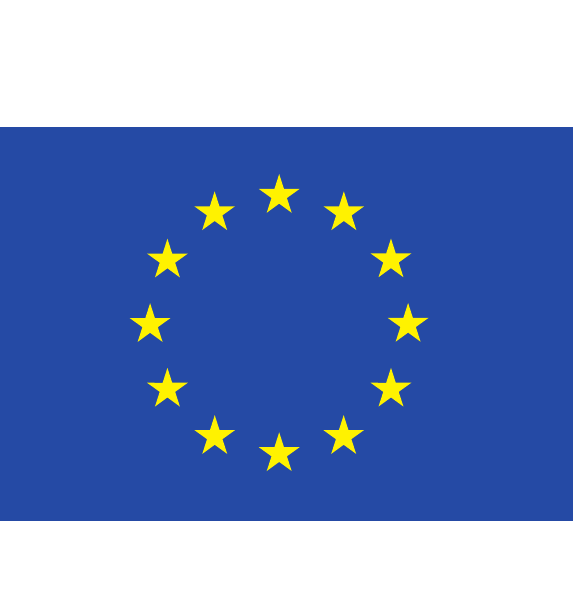}
 	\end{minipage}
 	\hspace{-3cm}
 	\begin{minipage}[l][1cm]{0.75\textwidth}
       	This work has received funding from the European Union's Horizon 2020
       	research and innovation programme under the Marie Sk\l{}odowska-Curie
       	grant agreement No 734922.
      \end{minipage}
 }
 \vspace{-0.9cm}

\begin{abstract}
Let $\mathcal{O}$ be a set of $k$ orientations in the plane, and let $P$ be a simple polygon in the plane. Given two points $p,q$ inside $P$, we say that $p$ $\mathcal{O}$-\emph{sees} $q$ if there is an $\mathcal{O}$-\emph{staircase} contained in $P$ that connects $p$ and~$q$. The \emph{$\mathcal{O}$-Kernel} of the polygon $P$, denoted by $\mathcal{O}$-$\Kl(P)$, is the subset of points of $P$ which $\mathcal{O}$-see all the other points in $P$. This work initiates the study of the computation and maintenance of $\mathcal{O}$-$\Kl(P)$ as we rotate the set $\mathcal{O}$ by an angle $\theta$, denoted by  $\mathcal{O}$-$\Kl_{\theta}(P)$. In particular, we consider the case when the set $\mathcal{O}$ is formed by either one or two orthogonal orientations, $\mathcal{O}=\{0^\circ\}$ or  $\mathcal{O}=\{0^\circ,90^\circ\}$. For these cases and $P$ being a simple polygon, we design efficient algorithms for computing the
$\mathcal{O}$-$\Kl_{\theta}(P)$ while $\theta$ varies in $[-\frac{\pi}{2},\frac{\pi}{2})$, obtaining: (i)~the intervals of angle~$\theta$ where $\mathcal{O}$-$\Kl_{\theta}(P)$ is not empty, (ii)~a value of angle~$\theta$ where $\mathcal{O}$-$\Kl_{\theta}(P)$ optimizes area or perimeter. Further, we show how the algorithms can be improved when $P$ is a simple orthogonal polygon. In addition, our results are extended to the case of a set $\mathcal{O}=\{\alpha_1,\dots,\alpha_k\}$.
\end{abstract}

\section{Introduction}

The problem of computing or reaching the kernel of a polygon is a well-known visibility problem in computational geometry~\cite{IK,LP,P2}, closely related to the problem of guarding a polygon~\cite{P,SRW,SW}, and also to robot navigation inside a polygon with the restriction that the robot path must be \emph{monotone} in some predefined set of orientations~\cite{HBZ19,SS97}. The present contribution goes a step further in the latter setting, allowing the polygon or, equivalently, the set of predefined orientations to rotate. Thus, we show how to compute the orientations that maximize the region from which every point can be reached following a monotone path.

A curve $\mathcal{C}$ is $0^\circ$-\emph{convex} if its intersection with any line parallel to the $x$-axis, called $0^\circ$-line, is connected (equivalently, if the curve $\mathcal{C}$ is $y$-monotone). Extending this definition, a curve $\mathcal{C}$ is $\alpha$-\emph{convex} if the intersection of $\mathcal{C}$ with any line forming a counterclockwise angle $\alpha$ with the positive $x$-axis, called $\alpha$-line, is connected (equivalently, if the curve $\mathcal{C}$ is monotone with respect to the direction $\alpha^\perp$).

Let us now consider a set $\mathcal{O}=\{\alpha_1,\dots,\alpha_k\}$ of $k$ orientations in the plane, each of them given by an oriented line~$\ell_i$, $1\le i\le k$, through the origin of the coordinate system and forming counterclockwise angle~$\alpha_i$ with the positive $x$-axis. Then, a curve is $\mathcal{O}$-\emph{convex} if it is $\alpha_i$-convex for all $i$, $1\le i\le k$, i.e., if the intersection of $\mathcal{C}$ with any line forming a counterclockwise angle $\alpha_i$, $1\le i\le k$, with the positive $x$-axis is connected
(equivalently, if it is
monotone with respect to all the directions~$\alpha_i^\perp$).
From now on, an $\mathcal{O}$-convex curve will be called an $\mathcal{O}$-\emph{staircase}. See Figure~\ref{FigureStaircases} for an illustration.

\begin{figure}[htb]
\centering
\includegraphics[width=0.5\textwidth]{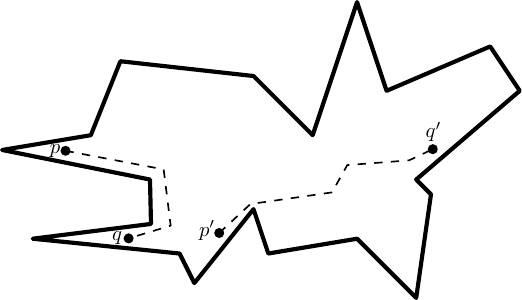}
\caption{A $\{0^\circ\}$-staircase which is not a $\{0^\circ,90^\circ\}$-staircase (left) and a $\{0^\circ,90^\circ\}$-staircase (right).}\label{FigureStaircases}
\end{figure}

Observe that the orientations in~$\mathcal{O}$ are between $0^\circ$ and $180^\circ$. Moreover, the only $[0^\circ,180^\circ)$-convex curves are lines, rays or segments. Throughout this paper, the angles of orientations in~$\mathcal{O}$ will be written in degrees, while the rest of angles will be measured in radians.

\begin{definition} \label{def:o-visible}
Let $p$ and $q$ be two points inside a simple polygon $P$. We say that $p$ and~$q$ $\mathcal{O}$-\emph{see} each other or, equivalently, that they are $\mathcal{O}$-\emph{visible} from each other, if there is an $\mathcal{O}$-staircase contained in $P$ that connects $p$ and $q$.
\end{definition}

In the example in Figure~\ref{FigureStaircases}, $p$ and $q$ are $\{0^\circ\}$-visible, while $p'$ and $q'$ are in addition $\{0^\circ,90^\circ\}$-visible. It is easy to see that $p$ and $q$ are not $\{90^\circ\}$-visible.

\begin{definition} \label{def:o-kernel}
The \emph{$\mathcal{O}$-Kernel} of $P$, denoted by $\mathcal{O}$-$\Kl(P)$, is the subset of points in $P$ which $\mathcal{O}$-see all the other points in $P$. The $\mathcal{O}$-Kernel of $P$ when the set $\mathcal{O}$ is rotated by an angle $\theta$ will be denoted by $\mathcal{O}$-$\Kl_{\theta}(P)$.
\end{definition}

\subsection{Previous related work}

Schuierer, Rawlins, and Wood~\cite{SRW} defined the restricted-orientation visibility or $\mathcal{O}$-\emph{visibility} in a simple polygon $P$ with $n$ vertices, giving an algorithm to compute the $\mathcal{O}$-$\Kl(P)$ in time $O(k+n\log k)$, with $O(k\log k)$ preprocessing time to sort the set $\mathcal{O}$ of $k$ orientations. In order to do so, they used the following observation.

\begin{observation}[\cite{SRW}]\label{obs:O-kernel}
For any simple polygon $P$, the $\mathcal{O}\textrm{-}\Kl(P)$ is $\mathcal{O}$-convex, connected, and
\[\mathcal{O}\textrm{-}\Kl(P)=\bigcap_{\alpha_i\in\mathcal{O}} \alpha_i\textrm{-}\Kl(P).\]
\end{observation}

The computation of the $\mathcal{O}$-$\Kl$ has been considered by Gewali~\cite{G} as well, who described an $O(n)$-time algorithm for orthogonal polygons without holes and an $O(n^2)$-time algorithm for orthogonal polygons with holes. The problem is a special case of the one considered by Schuierer and Wood~\cite{SW98} whose work implies an $O(n)$-time algorithm for orthogonal polygons without holes and an $O(n\log n+m^2)$-time algorithm for orthogonal polygons with $m\ge 1$ holes. More recently, Palios~\cite{P} gave an output-sensitive algorithm for computing the $\mathcal{O}$-$\Kl$ of an $n$-vertex orthogonal polygon $P$ with $m$ holes, for  $\mathcal{O}=\{0^\circ,90^\circ\}$; his algorithm runs in $O(n+m\log m+\ell)$ time, where $\ell\in O(1+m^2)$ is the number of connected components of $\{0^\circ,90^\circ\}$-$\Kl(P)$. Additionally, a modified version of this algorithm computes the number $\ell$ of connected components of the $\{0^\circ,90^\circ\}$-$\Kl$ in $O(n+m\log m)$ time~\cite{P}.

\subsection{Our contribution}

We consider the problem of computing and maintaining the $\mathcal{O}$-Kernel of $P$ while the set $\mathcal{O}$ rotates, that is, computing and maintaining $\mathcal{O}$-$\Kl_{\theta}(P)$ under variation of~$\theta$. For a simple polygon~$P$ and $\theta$ varying in~$[-\frac{\pi}{2},\frac{\pi}{2})$, we propose algorithms achieving the complexities in Table~\ref{Table:ResultsSimple}, where $\alpha(n)$ is the extremely-slowly-growing inverse of Ackermann's function \cite{A}. In addition, for the case of a simple orthogonal polygon~$P$, we propose improved algorithms to achieve the complexities in Table~\ref{Table:ResultsOrthogonal}. Note that looking for the minimum area or perimeter only makes sense where the kernel is non-empty.

\begin{table}[htb]
\centering
\renewcommand{\arraystretch}{1.1}
\setlength{\tabcolsep}{.13cm}
\footnotesize
\begin{tabular}{| c || c  c || c  c || c  c |}
\hline
\multirow{3}{*}{}                & \multicolumn{2}{|c||}{Get the intervals of~$\theta$ where} & \multicolumn{2}{|c||}{Get a value of~$\theta$ where} & \multicolumn{2}{|c|}{Get a value of~$\theta$ where} \\
& \multicolumn{2}{|c||}{the kernel is non-empty} & \multicolumn{2}{|c||}{the kernel has max/min area} & \multicolumn{2}{|c|}{the kernel has max/min perimeter} \\\cline{2-7}\rule{0pt}{10pt}
                                 & Time & Space & Time & Space & Time & Space \\

\hline
\multirow{2}{*}{$\{0^\circ\}$-$\Kl_{\theta}(P)$}  & $O(n\log n)$ & $O(n \alpha(n))$ & $O(n^2\alpha(n))$ & $O(n\alpha(n))$ & $O(n^2\alpha(n))$ & $O(n\alpha(n))$ \\
 & \multicolumn{2}{|c||}{(Theorem~\ref{thm:0-kernel-exists-simple})} & \multicolumn{2}{|c||}{(Theorem~\ref{thm:0-kernel-area-simple})} & \multicolumn{2}{|c|}{(Theorem~\ref{thm:0-kernel-perimeter-simple})} \\
\hline
\multirow{2}{*}{$\{0^\circ,90^\circ\}$-$\Kl_{\theta}(P)$}  & $O(n^2\alpha(n))$ & $O(n^2\alpha(n))$ & $O(n^2\alpha(n))$ & $O(n\alpha(n))$ & $O(n^2\alpha(n))$ & $O(n\alpha(n))$ \\
 & \multicolumn{2}{|c||}{(Theorem~\ref{thm:0-90-kernel-exists-simple})} & \multicolumn{2}{|c||}{(Theorem~\ref{teorema1x})} & \multicolumn{2}{|c|}{(Theorem~\ref{teorema1x})} \\
\hline
\multirow{2}{*}{$\mathcal{O}$-$\Kl_{\theta}(P)$}  & $O(kn^2\alpha(n))$
 & $O(kn^2\alpha(n))$ & $O(kn^2\alpha(n))$ & $O(kn\alpha(n))$ & $O(kn^2\alpha(n))$ & $O(kn\alpha(n))$ \\
 & \multicolumn{2}{|c||}{(Theorem~\ref{theorem5})} & \multicolumn{2}{|c||}{(Theorem~\ref{theorem7})} & \multicolumn{2}{|c|}{(Theorem~\ref{theorem7})} \\
\hline
\end{tabular}
\caption{Results for $P$ a simple polygon.
}
\label{Table:ResultsSimple}
\end{table}

\begin{table}[htb]
\centering
\renewcommand{\arraystretch}{1.1}
\setlength{\tabcolsep}{.13cm}
\footnotesize
\begin{tabular}{| c || c  c || c  c || c  c |}
\hline
\multirow{3}{*}{}                & \multicolumn{2}{|c||}{Get the intervals of~$\theta$ where} & \multicolumn{2}{|c||}{Get a value of~$\theta$ where} & \multicolumn{2}{|c|}{Get a value of~$\theta$ where} \\
& \multicolumn{2}{|c||}{the kernel is non-empty} & \multicolumn{2}{|c||}{the kernel  has max/min area} & \multicolumn{2}{|c|}{the kernel has max/min perimeter} \\\cline{2-7}\rule{0pt}{10pt}
                    & \phantom{........}Time & Space & \phantom{........}Time & Space & \phantom{...........}Time & Space \\

\hline
\multirow{2}{*}{$\{0^\circ\}$-$\Kl_{\theta}(P)$}  & \phantom{........}$O(n)$ & $O(n)$ & \phantom{........}$O(n)$ & $O(n)$ & \phantom{...........}$O(n)$ & $O(n)$ \\
 & \multicolumn{2}{|c||}{(Theorem~\ref{thm:0-kernel-exists-ortho})} & \multicolumn{2}{|c||}{(Theorem~\ref{thm:0-kernel-optareaperim-ortho})} & \multicolumn{2}{|c|}{(Theorem~\ref{thm:0-kernel-optareaperim-ortho})} \\
\hline
\multirow{2}{*}{$\{0^\circ,90^\circ\}$-$\Kl_{\theta}(P)$}  & \phantom{........}$O(n)$ & $O(n)$ & \phantom{........}$O(n)$ & $O(n)$ & \phantom{...........}$O(n)$ & $O(n)$ \\
 & \multicolumn{2}{|c||}{(Theorem~\ref{thm:0_90-kernel-exists-ortho})} & \multicolumn{2}{|c||}{(Theorem~\ref{thm:0-90-kernel-optareaperim-ortho})} & \multicolumn{2}{|c|}{(Theorem~\ref{thm:0-90-kernel-optareaperim-ortho})} \\
\hline
\multirow{2}{*}{$\mathcal{O}$-$\Kl_{\theta}(P)$}  & \phantom{........}$O(kn)$ & $O(kn)$ & \phantom{........}$O(kn)$ & $O(kn)$ & \phantom{...........}$O(kn)$ & $O(kn)$ \\
 & \multicolumn{2}{|c||}{(Theorem~\ref{thm:ortho_k_orientations})} & \multicolumn{2}{|c||}{(Theorem~\ref{thm:ortho_k_orientations})} & \multicolumn{2}{|c|}{(Theorem~\ref{thm:ortho_k_orientations})} \\
 \hline
\end{tabular}
\caption{Results for $P$ a simple orthogonal polygon.
} \label{Table:ResultsOrthogonal}
\end{table}

\section{The rotated \boldmath$\{0^\circ\}$-$\Kl_{\theta}(P)$ in a simple polygon $P$}

Let $(p_1,\ldots,p_n)$ be the counterclockwise sequence of vertices of a simple polygon $P$, which is considered to include its interior (sometimes called the \emph{body}). In this section we deal with the rotation of the set $\mathcal{O}=\{0^\circ\}$ by an angle $\theta\in [-\frac{\pi}{2},\frac{\pi}{2})$ and the computation of the corresponding $\mathcal{O}$-$\Kl_{\theta}(P)$, proving the results in the first row of Table~\ref{Table:ResultsSimple}.

\subsection{The $\{0^\circ\}$-$\Kl(P)$, its area, and its perimeter}

For the case $\mathcal{O}=\{0^\circ\}$ and $\theta=0$, i.e., for the $\{0^\circ\}$-$\Kl_{0}(P)$ or, more simply, $\{0^\circ\}$-$\Kl(P)$, the kernel is composed by the points inside $P$ which see every point in $P$ via a $y$-monotone curve. Note that if $P$ is a convex polygon, then the $\{0^\circ\}$-$\Kl(P)$ is the whole $P$. Schuierer, Rawlins, and Wood~\cite{SRW} presented the following definitions, observations, and results.

\begin{definition}\label{defi1}
A reflex vertex $p_i\in P$ is a \emph{reflex maximum} (respectively a \emph{reflex minimum}) if $p_{i-1}$ and $p_{i+1}$ are both below (resp.\ above) $p_i$. Analogously, a horizontal edge with two reflex vertices is a {\em reflex maximum} (resp.\ {\em minimum}) if its two neighbors are below (resp.\ above).
\end{definition}

Note that, throughout this work, the edges are considered to be closed and, therefore, containing their endpoints. Let $h_N$ be the horizontal line passing through a vertex~$p_N$ being a \emph{lowest reflex minimum} of~$P$ or, if $P$ does not have a reflex minimum, through the highest (convex) vertex of~$P$. Let $h_S$ be the horizontal line passing through a vertex~$p_S$ being a \emph{highest reflex maximum $p_S$} of~$P$ or, if $P$ does not have a reflex maximum, through the lowest (convex) vertex of~$P$. Let $S(P)$ be the strip defined by the horizontal lines $h_N$ and $h_S$, see Figure~\ref{Figure1}. Note that there are neither reflex minima nor maxima inside~$S(P)$.

\begin{lemma}[\cite{SRW}]
\label{lema1}
The $\{0^\circ\}$-$\Kl(P)$ is the region defined by the intersection $S(P)\cap P$.
\end{lemma}

\begin{corollary}[\cite{SRW}]
\label{corollary1}
The $\{0^\circ\}$-$\Kl(P)$ can be computed in $O(n)$ time.
\end{corollary}

Moreover, the horizontal lines $h_N$ and $h_S$ contain the segments of the \emph{north} boundary and of the \emph{south} boundary of the $\{0^\circ\}$-$\Kl(P)$; see again Figure~\ref{Figure1}. Lemma~\ref{lema1} is straightforward and Corollary~\ref{corollary1} is trivial by computing both the lowest reflex minimum and the highest reflex maximum in linear time and then computing $S(P)\cap P$ in additional linear time.

\begin{figure}[htb]
\centering
\includegraphics[width=0.45\textwidth]{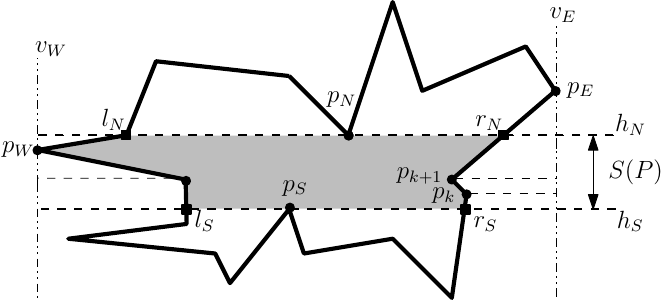}
\quad
\includegraphics[width=0.42\textwidth]{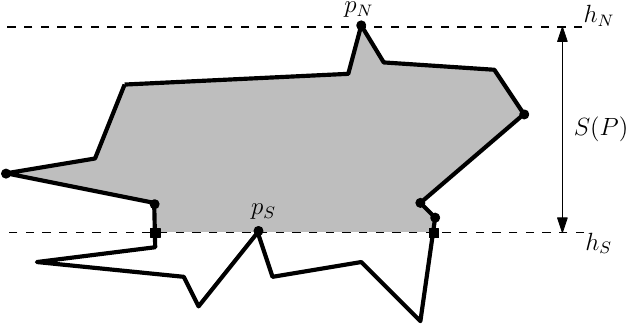}
\caption{Two examples of
$\{0^\circ\}$-$\Kl_{\theta}(P)$ for $\theta=0$. In the left example, the strip $S(P)$ is supported by a lowest reflex minimum~$p_N$ and a highest reflex maximum~$p_S$. In the right example there are no reflex minima and, therefore, the strip $S(P)$ is supported by the highest (convex) vertex~$p_N$ and the highest reflex maximum~$p_S$.}
\label{Figure1}
\end{figure}

Now, let $c^l$ and $c^r$ denote the \emph{left} and the \emph{right polygonal chains} defined, respectively, by those parts of the boundary of~$P$ which are inside $S(P)$. Let $|c^l|$ and $|c^r|$ denote their number of segments. It follows from the definition of $S(P)$ and Lemma~\ref{lema1} that both chains are $0^\circ$-convex curves, i.e., $y$-monotone chains; see Figure~\ref{Figure1} once more.

\begin{corollary}\label{cor9}
The area and the perimeter of the $\{0^\circ\}$-$\Kl(P)$ can be computed in $O(n)$ time.
\end{corollary}

\begin{proof}
To compute the area of the $\{0^\circ\}$-$\Kl(P)=S(P)\cap P$, we proceed as follows.
The area can be decomposed into (a finite number of) horizontal trapezoids defined by pairs of vertices in $c^l\cup c^r$ with consecutive $y$-coordinate. The area of these trapezoids can be computed in constant time, so the area of $\{0^\circ\}$-$\Kl(P)=S(P)\cap P$ can be computed in $O(|c^l|+|c^r|)$ time.

Computing the perimeter is even simpler, because we only need the addition of the lengths of $c^l$ and~$c^r$ plus the lengths of the north and south boundaries of the $\{0^\circ\}$-$\Kl(P)$, which can also be done in $O(|c^l|+|c^r|)$ time.
\end{proof}

\subsection{The existence of the $\{0^\circ\}$-$\Kl_{\theta}(P)$}
\label{sec:0-kernel-simple-polygons}

In this subsection, we show how to compute the intervals for~$\theta$ such that the $\{0^\circ\}$-$\Kl_{\theta}(P)$ is non-empty. First, we observe that we do not need a complete rotation, since $\{0^\circ\}$-$\Kl_{-\frac{\pi}{2}}(P)=\{0^\circ\}$-$\Kl_{\frac{\pi}{2}}(P)$. Also, notice that Definition~\ref{defi1}, for reflex maxima/minima with respect to the horizontal orientation, can be easily extended to any orientation $\theta\in [-\frac{\pi}{2},\frac{\pi}{2})$ as follows.

\begin{definition}\label{defi2}
A reflex vertex $p_i$ in a simple polygon $P$ where $p_{i-1}$ and $p_{i+1}$ are both below (respectively, above) $p_i$ with respect to a given orientation $\theta$ is a \emph{reflex maximum} (resp.\ a \emph{reflex minimum}) with respect to $\theta$. Analogously, an edge of angle~$\theta$ with two reflex vertices is a {\em reflex maximum} (resp.\ {\em minimum}) when its two neighbors are below (resp.\ above) with respect to the orientation~$\theta$.
\end{definition}

In order to know the intervals for~$\theta$ such that the $\{0^\circ\}$-$\Kl_{\theta}(P)$ is not empty, we need to maintain the boundary of the rotation by angle~$\theta$ of the strip~$S(P)$ previously defined, which will be denoted by~$S_{\theta}(P)$; see Figure~\ref{Figure2}. We need to extend Lemma~\ref{lema1} to any orientation~$\theta$:

\begin{lemma}
\label{rotatedLemma1}
The $\{0^\circ\}$-$\Kl_{\theta}(P)$ is the region defined by the intersection $S_{\theta}(P)\cap P$.
\end{lemma}

\begin{proof}
The claim follows from Lemma~\ref{lema1} and the fact that $\{0^\circ\}$-$\Kl_{\theta}(P)=\{0^\circ\}$-$\Kl(P_{\theta})$ and $S_{\theta}(P)=S(P_{\theta})$, where $P_{\theta}$ denotes the polygon $P$ rotated by the angle $\theta$. See Figure~\ref{Figure2}.
\end{proof}

\begin{figure}[htb]
\centering
\includegraphics[width=\textwidth]{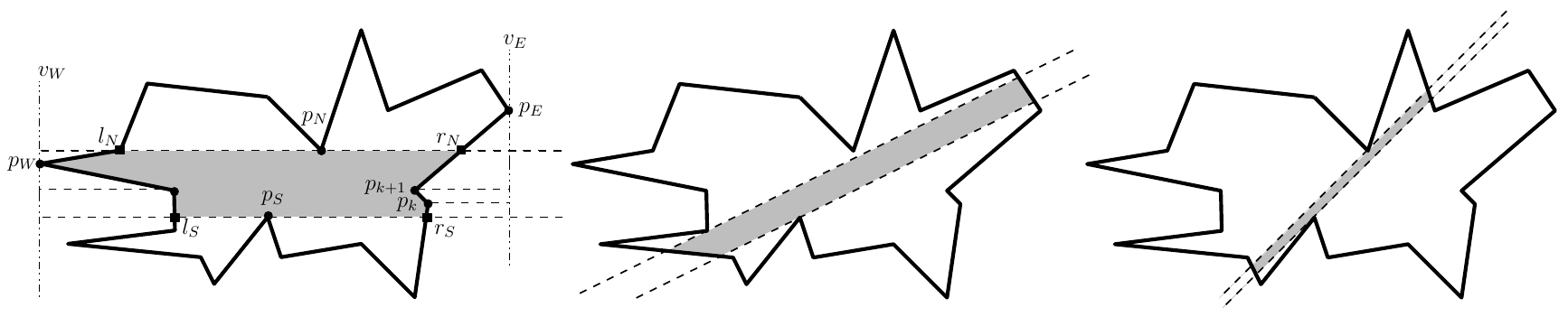}
\caption{A rotating $\{0^\circ\}$-$\Kl_{\theta}(P)$ for $\theta=0$ (left), $\theta={\frac{\pi}{8}}$ (middle), and $\theta={\frac{\pi}{4}}$ (right).}\label{Figure2}
\end{figure}

Now, we describe the main steps of our algorithm to compute the intervals of those values of $\theta$ within $[-\frac{\pi}{2},\frac{\pi}{2})$ such that $S_{\theta}(P)\not=\emptyset$ and, therefore, such that $\{0^\circ\}$-$\Kl_{\theta}(P)\not=\emptyset$.

\smallskip
\noindent{\bf Step 1: Angular intervals.}
For each vertex $p_i\in P$, if $p_i$ is reflex, we compute the angular intervals $[\theta_1^i,\theta_2^i)$ and $[\theta_1^i+\pi,\theta_2^i+\pi)$ of orientations~$\theta$ for which $p_i$ is a reflex maximum/minimum, defined when rotating the line containing the edge $p_{i-1}p_i$ up to the line containing the edge $p_ip_{i+1}$.
Otherwise, if $p_i$ is convex, we compute the angular intervals $[\theta_1^i,\theta_2^i)$ and $[\theta_1^i+\pi,\theta_2^i+\pi)$ of orientations~$\theta$ for which $p_i$ is the lowest/highest vertex of the rotated polygon $P_{\theta}$. Thus, in case that for some orientation~$\theta$ there is no reflex maximum/minimum, the lowest/highest convex vertex for that orientation will play the role of reflex maximum/minimum.
Note that an angular interval may be split into two, in case it contains the orientation~$\pi/2$.

\smallskip
\noindent{\bf Step 2: Dualization.} For the sake of efficiently handling the next step, we do the dualization of the set of vertices together with their relevant non-empty angular intervals from Step 1. The dualization function~$\ell$ we use is as follows: If $p=(a,b)$ is a point in the primal, its dual $\ell(p)$ is the line $\ell(p):\equiv y=ax-b$; if $r$ is the line given by $y=ax-b$ in the primal, its dual $\ell(r)$ is the  point $\ell(r):=(a,b)$. Moreover, the point $p=(a,b)$ lies below/on/above a line $l\equiv y=mx+c $ if and only if the line $\ell(p)\equiv  y=ax-b$ passes above/through/below the point $\ell(l)=(m,-c)$, see~\cite{BCKO}.

In this way, for a vertex $p_i\in P$ we translate the two lines which contain the incident edges $p_ip_{i-1}$ and $p_ip_{i+1}$ of the polygon $P$ into the corresponding dual points located on the dual line $\ell(p_i)$. In addition, we translate the set of lines through $p_i$ in the angular interval of $p_i$ into the corresponding set of dual points, which define a segment on the line $\ell(p_i)$. For an illustration, see the objects in red part in Figure~\ref{Figure3}. Thus, the angular interval of a point~$p_i$ is translated into the straight line segment on the line $\ell(p_i)$. Again, note that a vertex $p_i$ may contribute two segments in the dual plane, if the corresponding angular interval contains the orientation~$\pi/2$. The dualization process for all the other cases is done in an analogous way.

The dualization is performed as follows. On one hand, we dualize the reflex minima with their intervals which, in addition to the dual of intervals of the upper chain of the convex hull of~$P$, $\operatorname{CH}(P)$, (in blue in Figure~\ref{Figure3}) results in an arrangement  $\mathcal{D}_{\rm min}$ of line segments. On the other hand, we dualize the reflex maxima with their intervals (an example in red in Figure~\ref{Figure3}) which, together with the dual of the intervals of the lower chain
of~$\operatorname{CH}(P)$, gives an arrangement $\mathcal{D}_{\rm max}$ of line segments. Both arrangements have a linear number of line segments in the dual plane.

\begin{figure}[htb]
\centering
\includegraphics[width=0.85\textwidth]{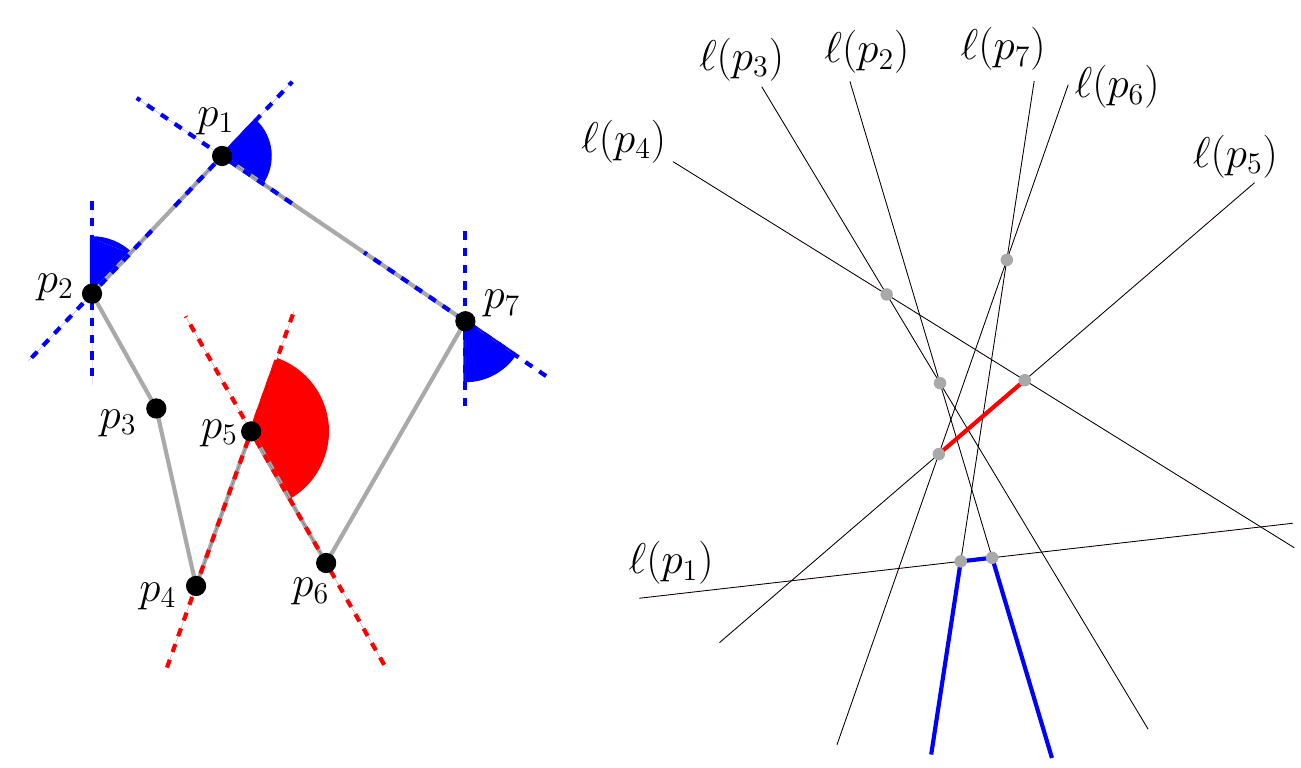}
\caption{In red, dualization of the angular interval corresponding to the vertex~$p_5$ (left) in the primal, which in the dual translates into a segment on the line~$\ell(p_5)$ (right). In blue, the angular intervals of the vertices $p_7,p_1,p_2$ in the upper chain of the convex hull in the primal (left), translate into the lower envelope of the arrangement in the dual (right).}\label{Figure3}
\end{figure}

\smallskip
\noindent{\bf Step 3: Event intervals.}
We compute the sequence of \emph{event intervals}, each of which is defined by a pair of orientation values $[\theta_1,\theta_2)\subset [-\frac{\pi}{2},\frac{\pi}{2})$ such that for any value $\theta\in [\theta_1,\theta_2)$, the strip $S_{\theta}(P)$ is supported by the same pair of vertices of $P$, in other words, such that the pair of vertices of~$P$ defining the lowest reflex minimum and the highest reflex maximum does not change for $\theta\in [\theta_1,\theta_2)$, recall Figure~\ref{Figure2}. In order to determine the sequence of event intervals, we exploit the following observation.

\begin{observation}\label{obs:x_theta}
The highest (resp. lowest) segment in $\mathcal{D}_{\rm min}$ (resp. $\mathcal{D}_{\rm max}$) intersected by the vertical line $x=\theta$ corresponds in the primal to the lowest reflex minimum (resp. the highest reflex maximum) with respect to the orientation $\theta$.
\end{observation}

\begin{proof}
It directly follows from the already mentioned fact that the dualization reverses the above-below relations between lines and/or points.
\end{proof}

Taking into account the above observation, we compute the upper envelope of $\mathcal{D}_{\rm min}$, denoted by $\mathcal{U_{\mathcal{D}_{\rm min}}}$, and the lower envelope of $\mathcal{D}_{\rm max}$, denoted by $\mathcal{L_{\mathcal{D}_{\rm max}}}$~\cite{H}. Next, by sweeping the arrangement $\mathcal{U_{\mathcal{D}_{\rm min}}} \cup \mathcal{L_{\mathcal{D}_{\rm max}}}$, we obtain the sequence of pairs ``lowest reflex minimum and highest reflex maximum'' for all the {event intervals} $[\theta_1,\theta_2)$, as $\theta$ varies in $[-\frac{\pi}{2},\frac{\pi}{2})$.

\smallskip
\noindent{\bf Step 4: Non-empty \boldmath{$\{0^\circ\}$-$\Kl_{\theta}(P)$}.} Recall that, by Lemma~\ref{rotatedLemma1}, the strip $S_{\theta}(P)$ is empty if, with respect to~$\theta$, the lowest reflex minimum is below the highest reflex maximum. Therefore, this step relies only on scanning the relevant pairs from Step 3 and checking whether the lowest reflex minimum is above the highest reflex maximum, which results in the angular intervals $[\theta_1,\theta_2)\subset [-\frac{\pi}{2},\frac{\pi}{2})$ such that $\{0^\circ\}$-$\Kl_{\theta}(P)\not=\emptyset$ for all the values of $\theta\in [\theta_1,\theta_2)$.

 \begin{algorithm}[ht!]
 \caption{\label{algo}Computing the intervals of $\theta$ such that $\{0^\circ\}$-$\Kl_{\theta}(P)\not=\emptyset$}
 \ \\ \
 \hglue 5pt \textbf{Input:} A simple polygon $P$ with $n$ vertices\\
 \hglue 5pt \textbf{Output:} Set $\cal I$ of event intervals for angles $\theta$ such that $\{0^\circ\}$-$\Kl_{\theta}(P)\not=\emptyset$
 \begin{algorithmic}[1]
     \Statex
     \Statex \hglue -6mm {\sc{\bf STEP 1:} Angular intervals}
     \For{$i = 1$ to $n$}
     \If{$p_i\in P$ is reflex}
     \State compute $[\theta_1^i,\theta_2^i)$ and $[\theta_1^i+\pi,\theta_2^i+\pi)$ such that $p_i$ is reflex maximum/minimum
     \EndIf
     \If{$p_i\in P$ is convex}
     \State compute $[\theta_1^i,\theta_2^i)$ and $[\theta_1^i+\pi,\theta_2^i+\pi)$ such that $p_i$ is the lowest/highest vertex of $P_{\theta}$,
     \State proceed like $p_i$ being a vertex reflex minimum/maximum
     \EndIf
     \EndFor
     \Statex
     \Statex \hglue -6mm {\sc{\bf STEP 2:} Dualization of vertices with their angular events from Step 1}
     \For{$i = 1$ to $n$}
     \If{$p_i$ is a reflex maximum
     }
     \State translate the angular interval of $p_i$ into the
     line segment on $\ell(p_i)$ and include this in an arrangement $\mathcal{D}_{\rm max}$
     \EndIf
     \If{$p_i$ is a reflex minimum
     }
     \State translate the angular interval of $p_i$ into the
     line segment on $\ell(p_i)$ and include this in an arrangement $\mathcal{D}_{\rm min}$
    \EndIf
    \State (Note that a reflex vertex may contribute two segments in the dual.)
    \EndFor
    \State Include in $\mathcal{D}_{\rm max}$ the dual of the lower chain of~$\operatorname{CH}(P)$ and include in $\mathcal{D}_{\rm min}$ the dual of the upper chain of~$\operatorname{CH}(P)$
    \Statex
    \Statex  \hglue -6mm {\sc{\bf STEP 3:} Event intervals}
    \State Compute the \emph{event intervals} such that $S_{\theta}(P)$ is supported by the same pair of vertices
    \State Compute the upper envelope $\mathcal{U_{\mathcal{D}_{\rm min}}}$ of $\mathcal{D}_{\rm min}$
    \State Compute the lower envelope $\mathcal{L_{\mathcal{D}_{\rm max}}}$ of $\mathcal{D}_{\rm max}$
    \State Sweep $\mathcal{U_{\mathcal{D}_{\rm min}}} \cup \mathcal{L_{\mathcal{D}_{\rm max}}}$ and compute the ``lowest reflex minimum and highest reflex maximum'' for the {event intervals}
    \Statex
    \Statex  \hglue -6mm {\sc{\bf STEP 4:} Non-empty $\{0^\circ\}$-$\Kl_{\theta}(P)$}
    \State Scan the vertex pairs from STEP 3, checking whether the lowest reflex minimum is above the highest reflex maximum and, if so, add the corresponding interval to an initially empty set~$\mathcal{I}$
    \State \textbf{output} $\cal I$
\end{algorithmic}
\end{algorithm}

\medskip\noindent\emph{Analysis of Algorithm~\ref{algo}}. The correctness of Algorithm~\ref{algo} follows from the discussion above, in particular from the concept of dualization together with Observation~\ref{obs:x_theta}. About the complexity, STEPS~1 and~2 can be done in linear time and space, in particular, by computing the convex hull of the simple polygon~$P$~\cite{McA79}. STEP~3 can be done in $O(n\log n)$ time, since the computation of the upper (and the lower) envelope of a set of $n$ possibly-intersecting straight-line segments can be done in $O(n\log n)$ time~\cite{H}. Finally, STEP~4 can be accomplished in $O(n\alpha(n))$, since the upper envelope and the lower envelope of a set of $n$ possibly-intersecting straight-line segments in the plane have worst-case size $O(n\alpha(n))$, where $\alpha(n)$ is the extremely-slowly-growing inverse of Ackermann's function~\cite{A}.

\begin{theorem}\label{thm:0-kernel-exists-simple}
For a simple polygon $P$ with $n$ vertices, the set of event intervals $[\theta_1,\theta_2)\subset [-\frac{\pi}{2},\frac{\pi}{2})$ such that $\{0^\circ\}$-$\Kl_{\theta}(P)\not=\emptyset$ for $\theta\in [\theta_1,\theta_2)$ can be computed
in $O(n\log n)$ time and $O(n\alpha(n))$ space.
\end{theorem}

\begin{proof}
The result is a direct consequence of applying Algorithm~\ref{algo}, whose correctness as well as time and space complexities follow from the analysis above.
\end{proof}

\subsection{Optimizing the area of the $\{0^\circ\}$-$\Kl_{\theta}(P)$}
\label{subsec:opt_area_0-kernel}

Let us consider the problem of optimizing the area of the $\{0^\circ\}$-$\Kl_{\theta}(P)$, i.e., computing the value(s) of~$\theta$ such that the area of $\{0^\circ\}$-$\Kl_{\theta}(P)$ is maximum or minimum (note that the latter only makes sense where the kernel is non-empty). The idea of our approach is based upon Algorithm~\ref{algo} for computing the set of event intervals $[\theta_1,\theta_2)\subset [-\frac{\pi}{2},\frac{\pi}{2})$ such that $\{0^\circ\}$-$\Kl_{\theta}(P)\not=\emptyset$ for all the values of $\theta\in [\theta_1,\theta_2)$ (Theorem~\ref{thm:0-kernel-exists-simple}). Namely, we do the following:

\smallskip
\noindent{\bf Step A: \boldmath{$\{0^\circ\}$-$\Kl_{\theta}(P) \not=\emptyset$}.} Run STEPS 1-4 of Algorithm~\ref{algo}.

\smallskip
\noindent{\bf Step B: Vertex events.}
For each event interval $[\theta_1,\theta_2)$ from Step~4 (within which the highest reflex maximum and the lowest reflex minimum do not change), we subdivide $[\theta_1,\theta_2)$ every time that, as $\theta$ varies, a vertex of the simple polygon $P$ either stops or starts contributing to the current boundary of the $\{0^\circ\}$-$\Kl_{\theta}(P)$. Observe that at, every such subdivision step, the differential in the area can be decomposed into triangles, as illustrated in Figure~\ref{Figure1-triangles}. In particular, for each of these consecutive subintervals $[\beta_j,\beta_{j+1})$ of $[\theta_1,\theta_2)$, we have:
\begin{equation}
Area (\{0^\circ\}\text{-}\Kl_{\beta}(P))= Area (\{0^\circ\}\text{-}\Kl_{\beta_{j}}(P))+ A_1(\beta) + A_2(\beta) - B_1(\beta) - B_2(\beta).
\label{eq:area}
\end{equation}
\noindent Thus, for such $\beta\in [\beta_j,\beta_{j+1})$, the area of the $\{0^\circ\}$-$\Kl_{\beta}(P)$ can be expressed, using simple trigonometric relations, as a function $A(\beta)$ of the angle of rotation $\beta\in [\beta_j,\beta_{j+1})$, as detailed in Section~\ref{appendix:area} in the appendix. Thus, it only remains to  obtain the maximum value of that function in the subinterval. In the mentioned Section~\ref{appendix:area} we show how this calculation is reduced to find the real solutions of a polynomial equation in $t$ of degree $6$. The final solution to the problem is then the best one over all those computed for these consecutive subintervals $[\beta_j,\beta_{j+1})$.

\begin{figure}[htb]
\centering
\includegraphics[width=0.6\textwidth]{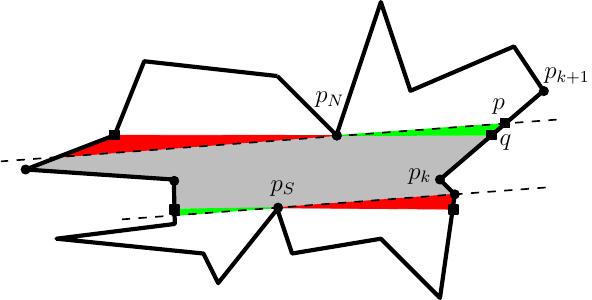}
\caption{The four triangles $A_1(\beta)$, $A_2(\beta)$ (in green), and $B_1(\beta)$, $B_2(\beta)$ (in red).}\label{Figure1-triangles}
\end{figure}

\smallskip
Clearly, Step B requires computing and maintaining the boundary of $\{0^\circ\}$-$\Kl_{\theta}(P)$, in particular, maintaining the set of vertices of the current left and right boundary chains, respectively denoted by $c^l_{\theta}$ and $c^r_{\theta}$, of $\{0^\circ\}$-$\Kl_{\theta}(P)$ as $\theta \in [\theta_1,\theta_2)$ varies (also for all the possible consecutive subintervals $[\beta_j,\beta_{j+1})$ of $[\theta_1,\theta_2)$); see again Figure~\ref{Figure1-triangles}. For this purpose, we compute the intersections of the lines $h_N(\theta)$ and $h_S(\theta)$ with the boundary of $P$, maintaining the information of the first and the last vertices of $c^l_{\theta}$ and $c^r_{\theta}$ in the current interval $[\theta_1,\theta_2)$. Now, as $\theta$ varies, the next vertex event can be computed in constant time by sweeping (and so modifying ad-hoc) chains $c^l_{\theta}$ and~$c^r_{\theta}$, in particular, using the circular order of the vertices of the polygon $P$ and taking the smallest among the relevant angles defined by the current line $h_N(\theta)$ (resp.\ $h_S(\theta)$), the point $p_N(\theta)$ (resp.\ $p_S(\theta)$),
and the relevant first polygon vertex on~$c^l_{\theta}$ and the first polygon vertex after the last polygon vertex on $c^r_{\theta}$ (resp. the first polygon vertex on $c^r_{\theta}$ and the first polygon vertex after the last polygon vertex on $c^l_{\theta}$).

\begin{algorithm}[ht!]
\caption{\label{algo2}Computing the maximum area of $\{0^\circ\}$-$\Kl_{\theta}(P)$}
\ \\ \
\hglue 5pt \textbf{Input:} A simple polygon $P$ with $n$ vertices\\
\hglue 5pt \textbf{Output:} An angle $\theta$ such that $Area(\{0^\circ\}$-$\Kl_{\theta}(P))$ is maximum and the maximum value of the area
\begin{algorithmic}[1]
    \Statex
    \Statex \hglue -6mm {\sc{\bf STEP A:}
    $\{0^\circ\}$-$\Kl_{\theta}(P) \not=\emptyset$}
    \State Run STEPS 1-4 from Algorithm~\ref{algo}
    \Statex
    \Statex \hglue -6mm {\sc{\bf STEP B:}
    Vertex events.}
    \For{each $[\theta_1,\theta_2)$ from STEP~4 of Algorithm~\ref{algo}}
    \If{a vertex of $P$ stops/starts appearing on the current boundary of the $\{0^\circ\}$-$\Kl_{\theta}(P)$}
    \State subdivide $[\theta_1,\theta_2)$ into consecutive subintervals $[\beta_j,\beta_{j+1})$ and decompose the differential of the area into triangles
    %\leonidas{perhaps change "decomposing...triangles" into: so that for any two angles in each subinterval, the differential in the area of the corresponding kernels consists of $4$ triangles as in Figure~\ref{Figure1-triangles}}\pawel{Good idea...} \carlos{No here, in the mid of the algorithm, better in the paragraph before, just after ... Figure 5 (...). For ...}
    \EndIf
    \EndFor
    \For{each subinterval  $[\beta_j,\beta_{j+1})$ and $\beta\in [\beta_j,\beta_{j+1})$}
    $$A(\beta)=
    Area(\{0^\circ\}\text{-}\Kl_{\beta}(P))= Area (\{0^\circ\}\text{-}\Kl_{\beta_{j}}(P))+ A_1(\beta) + A_2(\beta) - B_1(\beta) - B_2(\beta)$$
    \State Find the real solutions of a polynomial equation, and maintain the maximum value of $A(\beta)$ and the corresponding angle
    \EndFor
    \State \textbf{output} the maximum value of the area and the corresponding angle
\end{algorithmic}
\end{algorithm}

One can wonder whether the same vertex of a simple polygon $P$ may contribute to a vertex event for several event intervals. Surprisingly enough, there can be $\Theta(n)$ distinct vertices, each of them contributing $\Theta(n)$ vertex events, as illustrated in Figure~\ref{Figure4}. By Theorem~\ref{thm:0-kernel-exists-simple} we know that the number of event intervals is at most $O(n\alpha(n))$ thus, there may be as many as $O(n^2\alpha(n))$ vertex events (consecutive subintervals) involving in total $O(n^2\alpha(n))$ non-empty kernels $\{0^\circ\}$-$\Kl_{\theta}(P)$ having combinatorially different boundaries, implying the time complexity for computing the angle $\theta$ that maximizes (or minimizes) the area of $\{0^\circ\}$-$\Kl_{\theta}(P)$. To see this, it is enough to construct a simple polygon $P'$ by replicating the set of four points $\{p_1,p_2,p_3,p_4\}$ in Figure~\ref{Figure4} a linear number of times, and keeping the $\Theta(n)$ vertices in the corner. As we will see later, this bound also works for the computation of the maximum (or minimum) value of the perimeter of $\{0^\circ\}$-$\Kl_{\theta}(P)$. From this discussion we get the following result.

\begin{proposition}\label{propo1}
For a simple polygon $P$ with $n$ vertices, the number of vertex events or consecutive subintervals $[\beta_j,\beta_{j+1})$ where Algorithm~\ref{algo2} has to optimize the area of $\{0^\circ\}$-$\Kl_{\theta}(P)$ is $O(n^2\alpha(n))$.
\end{proposition}

\begin{proof}
The $O(n^2\alpha(n))$ bound comes from the simple polygon $P'$ constructed above based on Figure~\ref{Figure4}, taking into account the computation of the envelopes for obtaining the event intervals in Theorem~\ref{thm:0-kernel-exists-simple}.
\end{proof}

\begin{figure}[htb]
\centering
\includegraphics[width=0.6\textwidth]{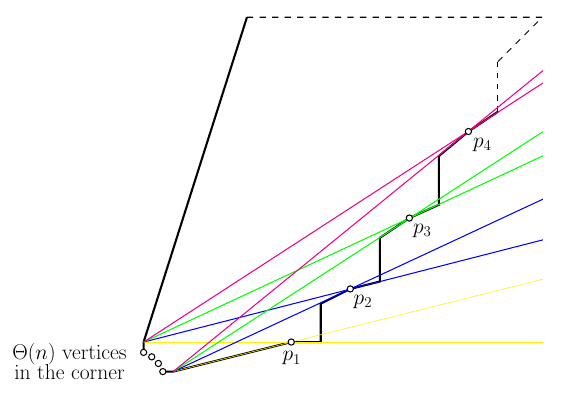}
\caption{For each vertex $p_i$, $1\le i\le 4$, all the $\Theta(n)$ vertices in the corner will be scanned again.}\label{Figure4}
\end{figure}

\noindent\emph{Analysis of Algorithm~\ref{algo2}}. The correctness of Algorithm~\ref{algo2} follows from the discussion above. Namely, STEP~A consists on running Algorithm~\ref{algo}, so it takes $O(n\log n)$ time and $O(n\alpha(n))$ space, obtaining $O(n\alpha(n))$ event intervals. By Proposition~\ref{propo1}, the number of vertex events or consecutive subintervals can be $O(n^2\alpha(n))$, and STEP~B spends constant time for the optimization in each of them, see Section~\ref{appendix:area}. Thus, this implies $O(n^2\alpha(n))$ time and $O(n\alpha(n))$ space in total. Notice that when we change from an event interval to the next event interval, we might have to manage a situation like the one illustrated in Figure~\ref{Figure4},
but this can be done in linear time and space since we translate one side of the kernel in parallel with the endpoints going through vertices on the boundary of $P$ (vertices in the  corner in Figure~\ref{Figure4}). Thus, it does not change the total time complexity because it implies an additional $O(n^2 \alpha(n))$ time;
also the space complexity does not change because the algorithm always reuses the linear space.

\begin{theorem}
\label{thm:0-kernel-area-simple}
For a simple polygon $P$ with $n$ vertices, an angle~$\theta$ that maximizes/minimizes the value of the area of $\{0^\circ\}$-$\Kl_{\theta}(P)$ can be computed in $O(n^2\alpha(n))$ time and $O(n\alpha(n))$ space.
\end{theorem}

\begin{proof}
Again, the correctness of our approach and the time and space complexities follow from the discussion above on the analysis of Algorithm~\ref{algo2} and Proposition~\ref{propo1}. The problem of minimizing the area, where meaningful, is handled in the same way.
\end{proof}

\subsection{Optimizing the perimeter of the $\{0^\circ\}$-$\Kl_{\theta}(P)$}
\label{subsec:opt_peri_0-kernel}

Consider now the problem of optimizing the perimeter of $\{0^\circ\}$-$\Kl_{\theta}(P)$, denoted by $\perim(\theta)$, where the goal is to compute the value(s) of $\theta$ such that $\perim(\theta)$ is maximum or minimum (note that the latter only makes sense where the kernel is non-empty). Observe that we can apply the same approach as the one proposed for optimizing the area of $\{0^\circ\}$-$\Kl_{\theta}(P)$ in Algorithm~\ref{algo2}, with the only difference that now, when handling the vertex events (defined and computed exactly in the same way as in the case of optimizing the area in Step B), we need to
handle the expression for the polygon perimeter. Clearly, the differential in the perimeter can be decomposed as adding two segments and subtracting two other segments, see again Figure~\ref{Figure1-triangles}, and thus the perimeter can then be expressed, using simple trigonometric relations, as a function $\Pi(\beta)$ of the angle of rotation $\beta\in[\beta_j,\beta_{j+1})$, see Section~\ref{appendix:perimeter} in the appendix. Then, it  only remains to obtain the maximum value of that function in the interval $[\beta_j,\beta_{j+1})$. As detailed in Section~\ref{appendix:perimeter}, this amounts to finding the real solutions of a polynomial equation in $t$ of constant degree. Consequently, we may conclude with the following result, where the minimization of the perimeter, if meaningful, is handled analogously.

\begin{theorem}\label{thm:0-kernel-perimeter-simple}
For a simple polygon $P$ with $n$ vertices, an angle~$\theta$ such that the value of the perimeter of $\{0^\circ\}$-$\Kl_{\theta}(P)$ is maximum/minimum can be computed in $O(n^2\alpha(n))$ time and $O(n\alpha(n))$ space.
\end{theorem}

\section{The rotated \boldmath$\{0^\circ,90^\circ\}$-$\Kl_{\theta}(P)$ of a simple polygon $P$}

We continue our study on the problem of computing the $\mathcal{O}$-$\Kl$ of a simple polygon $P$ considering the case when $\mathcal{O}$ is given by two perpendicular orientations which rotate simultaneously, for which we prove the results in the second row of Table~\ref{Table:ResultsSimple}. Notice that the two orientations do not need to be perpendicular for the proofs nor the algorithm in this section, because we are using Observation~\ref{obs:O-kernel}. Moreover, since the problem for a set $\mathcal{O}$ with $k$ orientations reduces to computing and maintaining the intersection of $k$ different kernels, the results in the third row of Table~\ref{Table:ResultsSimple} will follow as~well.

\subsection{The existence of the $\{0^\circ,90^\circ\}$-$\Kl_{\theta}(P)$}
\label{subsec:simple_non-empty_0-90}

Taking into account Observation~\ref{obs:O-kernel}, one can determine the $\{0^\circ,90^\circ\}$-$\Kl_{\theta}(P)$ by computing the intersection of the two kernels $\{0^\circ\}$-$\Kl_{\theta}(P)$ and $\{90^\circ\}$-$\Kl_{\theta}(P)$, respectively. Note that, in fact, the latter equals the $\{0^\circ\}$-$\Kl_{\theta+90^\circ}(P)$. In the following, the points $p_W(\theta)$ and $p_E(\theta)$ for the $\{90^\circ\}$-$\Kl_{\theta}(P)$ are analogous to the points $p_N(\theta)$ and $p_S(\theta)$ previously defined for the  $\{0^\circ\}$-$\Kl_{\theta}(P)$. Notice that $p_N(\theta+90^\circ)=p_W(\theta)$ and $p_S(\theta+90^\circ)=p_E(\theta)$,  recall Figure~\ref{Figure1}, and see Figure~\ref{Figure6}.
\begin{figure}[htb]
\centering
\centering
\subfloat{
\includegraphics[width=0.45\textwidth]{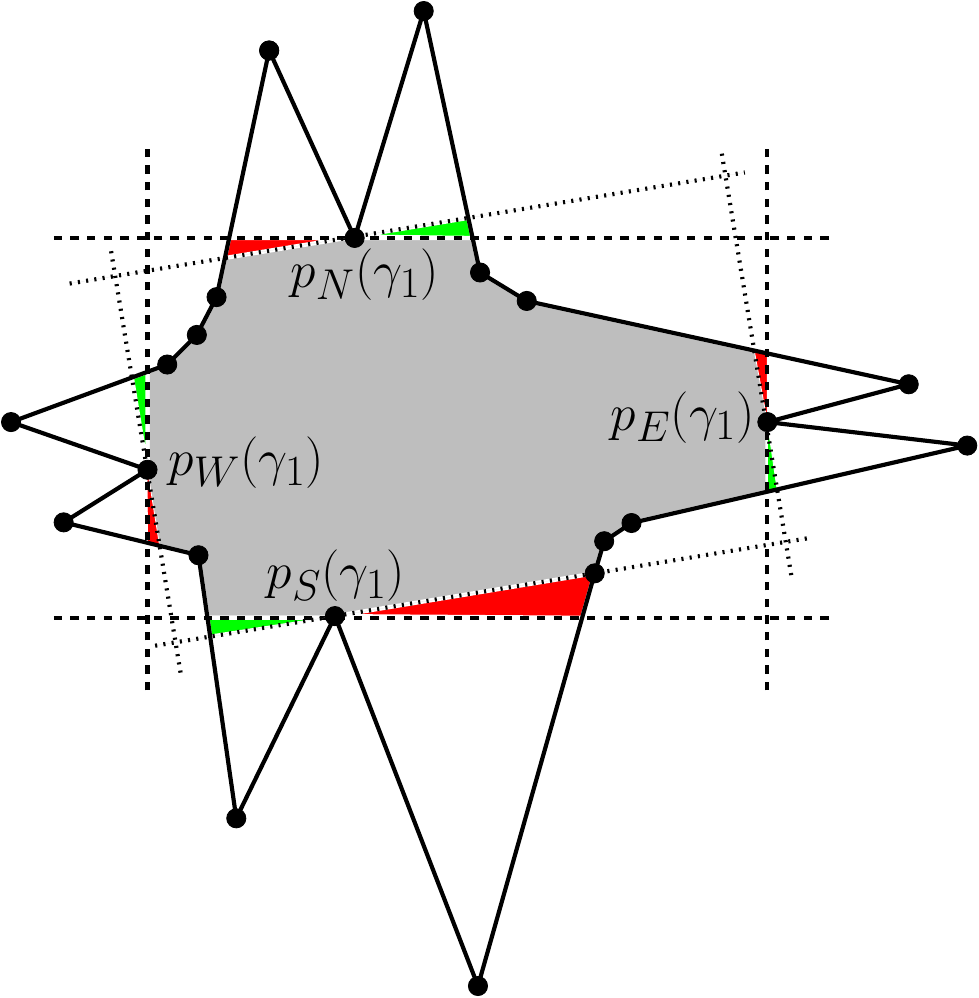}%
}
\qquad\qquad
\subfloat{
\includegraphics[width=0.4\textwidth]{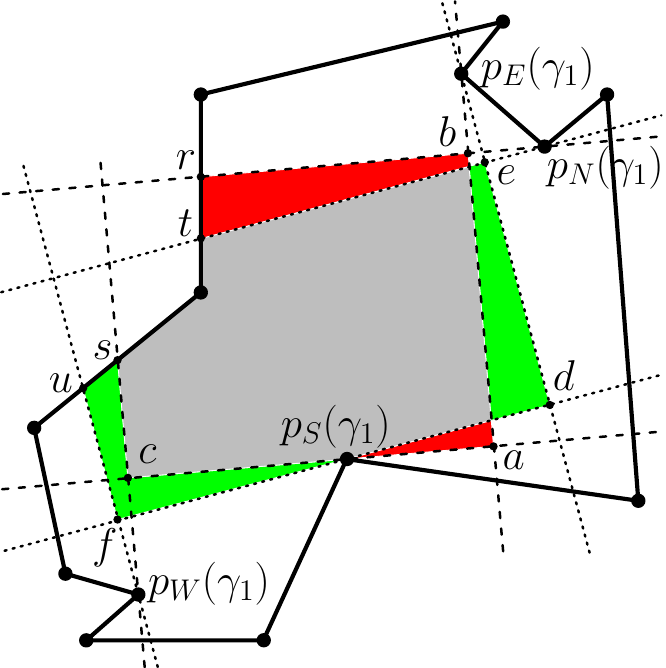}%
}
\caption{Left: A $\{0^\circ,90^\circ\}\textnormal{-}\Kl_{\theta}(P)$  and the rotated kernel in the next event, the area leaving (resp. entering) the kernel being depicted in red (resp. green). Right: A more general $\{0^\circ,90^\circ\}\textnormal{-}\Kl_{\theta}(P)$ and the rotated kernel in a slightly larger angle~$\beta$, depicting the entering and leaving areas as before. Note that, in both cases,  $p_i(\gamma_1)=p_i(\theta)$ for any $\theta \in [\gamma_1,\gamma_2)$, $i \in \{N,W,S,E\}$.
}
\label{Figure6}
\end{figure}

\subsubsection{Floating rectangle.}
Let $\theta \in [0,\pi/2)$ be an angle such that both $\{0^\circ\}$-$\Kl_{\theta}(P)$ and $\{90^\circ\}$-$\Kl_{\theta}(P)$ are non-empty. In what follows, we refer to the intersection $S_\theta(P) \cap S_{\theta+90^\circ}(P)$ as a {\em floating rectangle}, and denote it by $R_\theta$ (recall that $S_\alpha(P)$ denotes the strip defined by the lines $h_N(\alpha)$ and $h_S(\alpha)$ being, respectively, the line with slope~$\tan(\alpha)$ passing through $p_N(\alpha)$ and the line with slope $\tan(\alpha)$ passing through $p_S(\alpha)$). Clearly, by combining Lemma~\ref{lema1} with Observation~\ref{obs:O-kernel}, we observe that
\begin{equation} \label{eqn:rectangle}
\{0^\circ,90^\circ\}\textnormal{-}\Kl_{\theta}(P)=R_\theta \cap P,
\end{equation}
which immediately results in the following observation.

\begin{observation}\label{obs:empty_kernel_0_90}
The $\{0^\circ,90^\circ\}$-$\Kl_{\theta}(P)$ of a simple polygon $P$ is empty if, and only if:
\begin{itemize}
    \item[]
\begin{itemize}
    \item[{\em (A)}] either one of the two kernels $\{0^\circ\}$-$\Kl_{\theta}(P)$ or $\{90^\circ\}$-$\Kl_{\theta}(P)$ is empty, or
    \item[{\em (B)}] the floating rectangle $R_{\theta}$ is lying outside $P$ (as in Figure~\ref{Figure5}, left).
\end{itemize}
\end{itemize}
\end{observation}

\begin{figure}[htb]
\centering
\includegraphics[width=0.95\textwidth]{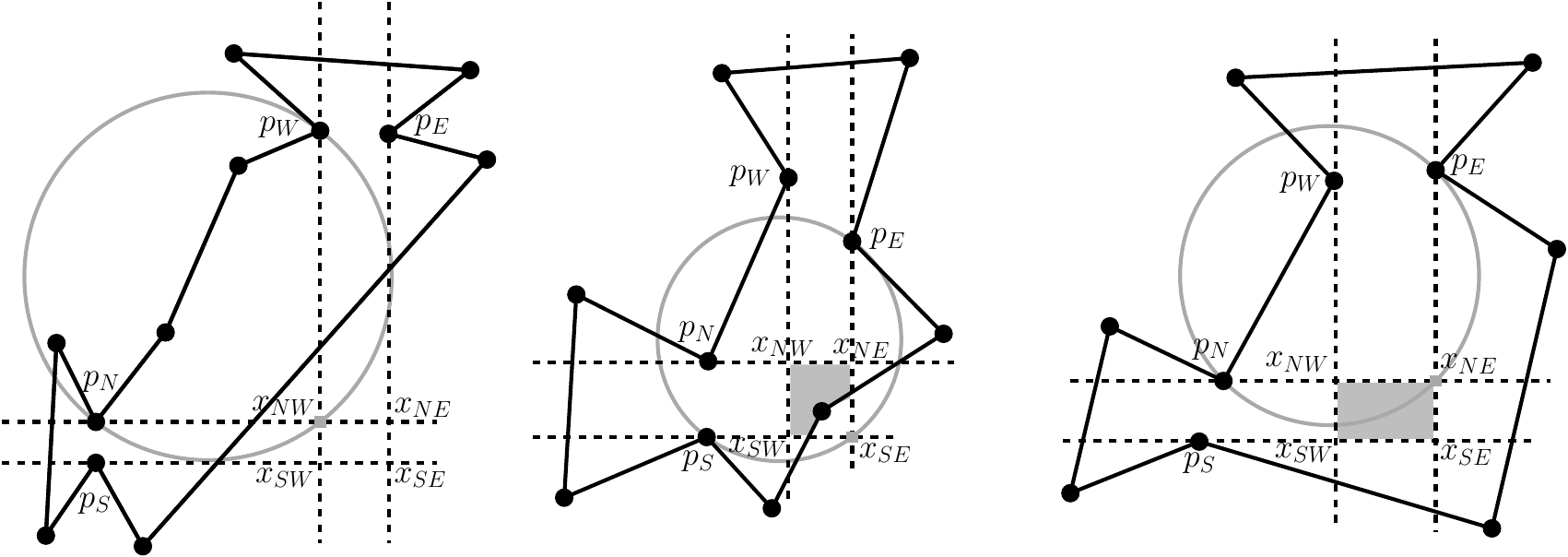}
\caption{Three types of kernel with the arcs for the vertices of the floating rectangle. Note that the references to the angle~$\theta$ in $p_i(\theta)$ and $x_{ij}(\theta)$ have been removed for the sake of an easier visualization.}\label{Figure5}
\end{figure}

Assume now that, following our approach proposed for the proof of Theorem~\ref{thm:0-kernel-exists-simple}, we have already computed the sequence ${\cal I}_{0^\circ}$ of event intervals where $\{0^\circ\}$-$\Kl_{\theta}(P)\not=\emptyset$, and in an analogous way the sequence ${\cal I}_{90^\circ}$ of event intervals where $\{90^\circ\}$-$\Kl_{\theta}(P)\not=\emptyset$. Now, in $O(n\alpha(n))$ time and space, we obtain from these two event sequences the sequence ${\cal I}$ (with complexity $O(n\alpha(n)$) of the event intervals corresponding to the simultaneous rotation of both kernels, saving only those non-empty intersections $I' \cap I''$ of event intervals $I' \in {\cal I}_{0^{\circ}}$ and $I'' \in {\cal I}_{90^{\circ}}$, where both kernels are non-empty. Once we have stored this data, as a matter of fact, we have handled Case (A) in Observation~\ref{obs:empty_kernel_0_90}.

Next, as regards Case (B) in Observation~\ref{obs:empty_kernel_0_90}, the following lemma allows us to check whether the intersection $\{0^\circ\}$-$\Kl_{\theta}(P)\cap \{90^\circ\}$-$\Kl_{\theta}(P)=\{0^\circ,90^\circ\}$-$\Kl_{\theta}(P)$ is non-empty.

\begin{lemma}\label{lemma:obs13}
Consider an event interval $[\gamma_1,\gamma_2) \in {\cal I}$ and an angle $\theta \in [\gamma_1,\gamma_2)$. Then the $\{0^\circ,90^\circ\}$-$\Kl_{\theta}(P)$ is non-empty in the following cases:
\begin{itemize}
\item[]
\begin{itemize}
    \item[{\em (B.1)}] At least one point among the current $p_N(\gamma_1)$, $p_S(\gamma_1)$, $p_E(\gamma_1)$, $p_W(\gamma_1)$ belongs to the floating rectangle~$R_{\gamma_1}$ $($see Figure~\ref{Figure6}$)$.
    \item[{\em (B.2)}] The polygon~$P$ contains at least one of the corners of the floating rectangle $R_\theta$ $($see Figure~\ref{Figure5}$)$.
\end{itemize}
\end{itemize}
\end{lemma}

\begin{proof}
First, if at least one of the cases (B.1), (B.2) holds then the $\{0^\circ,90^\circ\}$-$\Kl_{\gamma_1}(P)$ is non-empty. Assume now that the kernel is non-empty and suppose, for contradiction, that neither (B.1) nor (B.2) holds. Then, the fact that (B.2) does not hold implies that all $4$ corners of the rectangle~$R_{\gamma_1}$ lie outside $P$. Consider two adjacent corners $r,r'$ of $R_{\gamma_1}$ lying on the line~$h_N(\gamma_1)$ that goes through $p_N(\gamma_1)$. The fact that (B.1) does not hold implies that $p_N(\gamma_1)$ does not belong to the line segment connecting $r, r'$, But then, if there were a point~$q \in P$ on the segment~$r r'$, then the definition of the $\{0^\circ\}$-$\Kl_{\theta}(P)$ (see  Definition~\ref{def:o-kernel}) implies that $q$ should be $\gamma_1$-visible from $p_N(\gamma_1)$. Then, Definition~\ref{def:o-visible} implies that there is a $\gamma_1$-staircase~$C$ in~$P$ connecting $p_N(\gamma_1)$ and $q$; this is a contradiction because the intersection of $C$ with the line~$h_N(\gamma_1)$ which has slope $\tan(\gamma_1)$ is not connected. Thus, the entire edge~$r r'$ of $R_{\gamma_1}$ lies outside $P$.

Similarly, the other edges of $R_{\gamma_1}$ lie outside $P$ as well. Then, for any point~$q'$ inside $R_{\gamma_1}$, we can apply the same argument by using a line parallel to $h_N(\gamma_1)$ that goes through $q'$ (note that such a line intersects the strip~$S_{\gamma_1}$), proving that $q' \not\in P$. Therefore, the entire $R_{\gamma_1}$ lies outside $P$, in contradiction to the fact that the $\{0^\circ,90^\circ\}$-$\Kl_{\gamma_1}(P)$ is non-empty.
\end{proof}

Clearly, Case (B.1) can be checked in constant time, by the orientation test with the point considered and the two lines forming the relevant strip.
Notice that the situation of these four points cannot change during the event interval $[\gamma_1,\gamma_2)$, since $p_i(\gamma_1)=p_i(\theta)$ for any $\theta \in [\gamma_1,\gamma_2)$, $i \in \{N,W,S,E\}$.

\subsubsection{Arc events.}

For $i\in\{N,S\}$, let $x_{iW}(\theta)$ (resp.\ $x_{iE}(\theta)$) denote the intersection point of the line $h_i(\theta)$ with the line $h_N(\theta+90^\circ)$ (resp.\ $h_S(\theta+90^\circ)$), see Figure~\ref{Figure5}. In other words, the points $x_{ij}(\theta)$, $i\in\{N,S\}$ and $j\in\{W,E\}$,  are the relevant four corners of the floating rectangle $R_\theta$. Next, for $i\in\{N,S\}$ and $j\in\{W,E\}$, let $C_{ij}(\theta)$ denote the circle passing through the points $p_i(\theta),p_j(\theta)$ and $x_{ij}(\theta)$, again see Figure~\ref{Figure5}. Finally, let $\Breve{a}_{ij}(\theta)$ denote the arc of $C_{ij}(\theta)$ between $p_i(\theta)$ and~$p_j(\theta)$ such that $x_{ij}(\theta)$ belongs to $\Breve{a}_{ij}(\theta)$. Notice that the angle between points $p_i(\theta), x_{ij}(\theta)$ and $p_j(\theta)$ is the right angle, and so the point $x_{ij}(\theta)$ describes the semicircle having as diameter the segment $\overline{p_i(\theta)p_j(\theta)}$ (see again see Figure~\ref{Figure5}), thus impyling $C_{ij}(\theta)=C_{ij}(\gamma_1)$ and $\Breve{a}_{ij}(\theta)=\Breve{a}_{ij}(\gamma_1)$ for any $\theta \in [\gamma_1,\gamma_2)$.
Consequently, as $\theta$ varies in $[\gamma_1,\gamma_2)$, the point $x_{ij}(\theta)$
continuously moves along the arc $\Breve{a}_{ij}(\gamma_1)$. Moreover, we have the following observation.

\begin{observation}\label{obs:arc-event}
As $\theta$ varies in $[\gamma_1,\gamma_2)$, the point $x_{ij}(\theta)$
can change several times from the exterior to the interior of the polygon $P$ or vice versa.
\end{observation}

The claim follows from the interval $[\gamma_1,\gamma_2)$ being the intersection of event intervals and the fact that the boundary of the simple polygon $P$ can be a polyline of size $\Theta(n)$, as the one in Figure~\ref{Figure4}.
Taking into account Observation~\ref{obs:arc-event}, for an event interval $[\gamma_1,\gamma_2)$,
we can handle the case (B.2) in linear time. Because there are at most $O(n\alpha(n))$ event intervals, the total complexity for this step will be $O(n^2\alpha(n))$. Therefore, we can outline Algorithm~\ref{algo3}.

\begin{algorithm}[ht!]
\caption{\label{algo3}Computing the intervals of $\theta$ such that
$\{0^\circ,90^\circ\}$-$\Kl_{\theta}(P)\not=\emptyset$}
\ \\ \
\hglue 5pt \textbf{Input:} A simple polygon $P$ with $n$ vertices\\
\hglue 5pt \textbf{Output:} Sequence $\mathcal{E}$ of intervals for angles $\theta$ such that
$\{0^\circ,90^\circ\}$-$\Kl_{\theta}(P)\not=\emptyset$
\begin{algorithmic}[1]
    \Statex
    \Statex \hglue -6mm {\sc{\bf STEP I:} Event intervals: Checking Case (A) in Observation~\ref{obs:empty_kernel_0_90}}
    \State Apply Algorithm~\ref{algo} to compute the sequence ${\cal I}_{0^\circ}$ of event intervals where $\{0^\circ\}$-$\Kl_{\theta}(P)\not=\emptyset$
    \State Apply Algorithm~\ref{algo} to compute the sequence ${\cal I}_{90^\circ}$ of event intervals where $\{90^\circ\}$-$\Kl_{\theta}(P)\not=\emptyset$
    \State Combine ${\cal I}_{0^\circ}$ and ${\cal I}_{90^\circ}$ into the sequence ${\cal I} = \{I' \cap I'' = [\gamma_j,\gamma_{j+1})\; | \; I' \in {\cal I}_{0^{\circ}}, I'' \in {\cal I}_{90^{\circ}}\}$
    \Statex
    \Statex \hglue -6mm {\sc{\bf STEP II:}   Floating rectangle: Checking Case (B)}
    \For{each event interval $[\gamma_1,\gamma_{2})\in {\cal I}$ }
    \If{$p_N(\gamma_1)$ or $p_S(\gamma_1)$ or $p_E(\gamma_1)$ or $p_W(\gamma_1)$ belongs to $R_{\gamma_1}$ }
    \State Case (B.1) holds and $\{0^\circ,90^\circ\}$-$\Kl_{\theta}(P)\not=\emptyset$ for $\theta\in [\gamma_1,\gamma_{2})$. Insert $[\gamma_1,\gamma_{2})$ in an initially empty sequence $\mathcal{E}$
    %\EndIf
    \Else{
    \For{each vertex event $[\beta_1,\beta_{2})\subseteq [\gamma_1,\gamma_{2})$}
    \If{$x_{ij}(\theta)$, $i\in\{N,S\}, \, j\in\{W,E\}$, on $\Breve{a}_{ij}(\gamma_1)$ as $\theta\in [\beta_1,\beta_{2})$, is in the interior of $P$}
    \State Case (B.2) holds and $\{0^\circ,90^\circ\}$-$\Kl_{\theta}(P)\not=\emptyset$ for $\theta\in [\beta_1,\beta_{2})$. Insert $[\beta_1,\beta_{2})$ in $\mathcal{E}$
    \EndIf
    \EndFor
    }
    \EndIf
    \EndFor
\State {\bf output} $\mathcal{E}$
\end{algorithmic}
\end{algorithm}

\noindent\emph{Analysis of Algorithm~\ref{algo3}.} For STEP I we only need to apply twice Algorithm~\ref{algo}, and then do a refinement of two sequences of sizes $O(n\alpha(n))$, getting a sequence of size $O(n\alpha(n))$ in $O(n\log n)$ time and $O(n\alpha(n))$ space. STEP II has two
cases: Case~(B.1) takes only constant time to check whether some of the points belongs to the floating rectangle, and it is done $O(n\alpha(n))$ times, giving $O(n\alpha(n))$ total time complexity. Case ~(B.2) is also done $O(n\alpha(n))$ times but in each of them, we might have to check (in constant time) at most a linear number of vertex events or consecutive subintervals for each of the four vertices of the current floating rectangle. Therefore the total complexities of STEP II are $O(n^2\alpha(n))$ time and space. Notice that the space complexity of the algorithm is $O(n^2\alpha(n))$ because we are storing a sequence $\mathcal{E}$ of (possible) size $O(n^2\alpha(n))$.

\begin{theorem}\label{thm:0-90-kernel-exists-simple}
For a simple polygon $P$ with $n$ vertices, the sequence of consecutive intervals for the angles $\theta$ such that
$\{0^\circ,90^\circ\}$-$\Kl_{\theta}(P)\not=\emptyset$ can be computed in $O(n^2\alpha(n))$ time and space.
\end{theorem}

\begin{proof} The discussion above and the analysis of the complexities in Algorithm~\ref{algo3} provide the proof of this theorem.
\end{proof}

\subsubsection{Generalization to $k$ orientations}\label{subsec:k-orient}

One can extend Theorem~\ref{thm:0-90-kernel-exists-simple} to the case of a set $\mathcal{O}=\{\alpha_1,\dots,\alpha_k\}$ of $k$ orientations. In particular, Lemma~\ref{lemma:obs13} can be extended as follows. Instead of the four points $p_N(\gamma_1)$, $p_
S(\gamma_1)$, $p_E(\gamma_1)$, and $p_W(\gamma_1)$, we have $2k$ highest/lowest maximum/minimum reflex vertices according to the $k$ different orientations. The extended version of Condition (B.1) requires at least one of them to be inside the convex polygon defined by the intersection of the $k$ strips, what can be checked in $O(k)$ time and space, whereas Condition (B.2) holds if at least one vertex of this convex polygon is inside $P$, what can be checked in $O(kn^2\alpha (n))$ time and space.
Thus, we get the following result.

\begin{theorem}\label{theorem5}
For a simple polygon $P$ with $n$ vertices, the sequence of consecutive intervals for the angles $\theta$ such that $\{\alpha_1,\dots,\alpha_k\}$-$\Kl_{\theta}(P)\not=\emptyset$ can be computed in $O(kn^2\alpha(n))$ time and space. \end{theorem}

\subsection{Optimizing the area and perimeter of $\{0^\circ,90^\circ\}$-$\Kl_{\theta}(P)$ of simple polygons}
\label{subsec:simple_area-perimeter_0-90}

Let us consider a $\{0^\circ,90^\circ\}$-$\Kl_{\theta}(P)$ at some angle~$\theta = \gamma_1$ and suppose that the orientations are rotated to a slightly larger angle~$\beta$ so that the kernels at angle~$\gamma_1$ and $\beta$ are defined by the same reflex minima and maxima $p_i(\gamma_i)$, $i \in \{N,W,S,E\}$, and are bounded by the same edges of the polygon. The differential in the area of the kernels in the case shown in Figure~\ref{Figure6}, left, can be expressed in terms of $8$ triangles similar to the ones we saw for the $\{0^\circ\}$-$\Kl_{\theta}(P)$. The case we show in Figure~\ref{Figure6}, right, is more general and we have (for simplicity, we use here $p_i$ instead of $p_i(\gamma_i)$ for $i \in \{N,W,S,E\}$):
\begin{eqnarray} \label{eqn:0_90_area}
A(\beta) = A(\gamma_1)
& + & (A_T(p_S \,d \,p_E) - A_T(p_S \,a \,p_E))
+ (A_T(p_E \,e \,p_N) - A_T(p_E \,b \,p_N))\cr
& - & A_T(p_N \,r \,t)
+ A_T(p_E \,s \,u)
- (A_T(p_E \,f \,p_S) - A_T(p_E \,c \,p_S)),
\end{eqnarray}
where by $A_T(a \,b \,c)$ we denote the area of the triangle with vertices $a,b,c$. Thus, the differential in the area can be expressed using the area of at most $8$ triangles with 1 edge on a polygon edge and at most $4$ differences of two triangles with common base and whose third vertex moves along a circular arc.
The differential in the perimeter is (see Figure~\ref{Figure6}, right):
\begin{eqnarray} \label{eqn:0_90_perim}
\Pi(\beta) = \Pi(\gamma_1)
& + & (\Pi_T(p_S \,d \,p_E) - \Pi_T(p_S \,a \,p_E))
- (\Pi_T(p_E \,e \,p_N) - \Pi_T(p_E \,b \,p_N))\cr
& + & \Delta\Pi^-_T(p_N \,r \,t)
+ \Delta\Pi^+_T(p_E \,s \,u)
+ (\Pi_T(p_E \,f \,p_S) - \Pi_T(p_E \,c \,p_S))
\end{eqnarray}
where $\Pi_T(a \,b \,c)$ is the perimeter of the triangle with vertices $a,b,c$ and $\Delta\Pi^+_T(a \,b \,c)$
(resp. $\Delta\Pi^-_T(a \,b \,c)$) is the sum (resp. difference) of the difference of the lengths of the edges at angle~$\beta$ and $\gamma_1$ plus (resp. minus) the length of the third edge.

To compute and maintain the optimal values for the area and perimeter of the $\{0^\circ,90^\circ\}$-$\Kl_{\theta}(P)$, we can use the data computed in Section~\ref{subsec:simple_non-empty_0-90} about the intervals where this kernel is non-empty. Moreover, we can assume that in each of these intervals there are neither changes in the points of $P$ defining the kernel, nor changes in the vertices of the intersection rectangle of the two kernel strips. In particular, following Lemma~\ref{lemma:obs13} and Observation~\ref{obs:arc-event}, we compute different intervals for the cases when one, two, three, or the four vertices of the rectangle lie inside the kernel. This only implies a multiplicative constant factor in the number of event intervals. Thus, again a total of $O(n^2\alpha(n))$ intervals arise.

Next, we can analyze the method and formulas to compute the area or the perimeter according to the different types of intervals. We can always assume that we have computed the area or the perimeter of the previous interval, i.e., if we are going to analyze the interval $[\gamma_1,\gamma_2)$, then we know the values of the area and the perimeter for the previous interval $[\gamma'_1,\gamma'_2)$.

Thus, for the area or perimeter in Case (B.1) of Lemma~\ref{lemma:obs13}, if these four points are inside the kernel as illustrated in Figure~\ref{Figure6}, left, then we have to consider the $8$ triangles involved with the formulas for the area or perimeter, in an analogous way as for the case of one orientation $\{0^\circ\}$-$\Kl_{\theta}(P)$ in Subsections~\ref{subsec:opt_area_0-kernel} and~\ref{subsec:opt_peri_0-kernel}. If there are three, two, or only one of the points inside the kernel, it is enough to incorporate the corresponding new formulas for these cases. For the sake of easier reading, and since the complexity of the algorithm does not increase, the details for those cases are omitted.

An analogous situation arises for Case (B.2) of Lemma~\ref{lemma:obs13}: If all four rectangle corners are inside the polygon $P$, then it is easy to describe the formulas for the area and perimeter. We would have to add new formulas for the cases where there are three, two, or only one corner of the rectangle, but again the complexity of the algorithm does not change and details are omitted.

Thus, it is clear that the relevant issue for the algorithms optimizing area or perimeter is the total time for computing all of the $O(n^2\alpha(n))$ intervals (each one of them can be handled in constant time), which is $O(n^2\alpha(n))$. The space complexity is $O(n\alpha(n))$ because we only maintain the maximum/minimum values of the area or the perimeter but no all of the computed values, thus the used space is essentially for computing the set of event intervals. Notice that when we change from an event interval to the next event interval, we may have to manage a situation like the one illustrated in Figure~\ref{Figure4} from Proposition~\ref{propo1}, but this can be done in linear time and space since we translate one side of the kernel in parallel with a endpoint going through vertices on the boundary of $P$, and it does not change the total time and space complexities. Therefore, we have the following result.

\begin{theorem}\label{teorema1x}
For a simple polygon $P$ with $n$ vertices, an angle~$\theta$ such that the area or the perimeter of the $\{0^\circ,90^\circ\}$-$\Kl_{\theta}(P)$ are maximum/minimum can be computed in $O(n^2\alpha(n))$ time and $O(n\alpha(n))$ space.
\end{theorem}

\subsubsection{Generalization to $k$ orientations}

In a similar way as above, we can extend Theorem~\ref{teorema1x} to the case of a set $\mathcal{O}=\{\alpha_1,\dots,\alpha_k\}$ of $k$ orientations. Thus, we get the following result\footnote{
Actually, we can compute all angles $\theta$ maximizing/minimizing the area/perimeter in of $O(kn^2\alpha(n))$ space.
}.

\begin{theorem}\label{theorem7}
For a simple polygon $P$ with $n$ vertices, an angle~$\theta$ such that the area or the perimeter of
$\{\alpha_1,\dots,\alpha_k\}$-$\Kl_{\theta}(P)$ are maximum/minimum can be computed in $O(kn^2\alpha(n))$ time and $O(kn\alpha(n))$ space.
\end{theorem}

\section{Simple orthogonal polygons}\label{sec:ortho}

In this section, we confine our study to simple orthogonal polygons, showing how the results in Table~\ref{Table:ResultsSimple} can be improved to those in Table~\ref{Table:ResultsOrthogonal} for this case.

Each edge of an orthogonal polygon is a {\em {\rm N}-edge}, {\em {\rm S}-edge}, {\em {\rm E}-edge}, or {\em {\rm W}-edge}
depending on whether it bounds the polygon from the north, south, east, or west,
respectively. In particular, for $D \in \{{\rm N}, {\rm S}, {\rm E}, {\rm W}\}$, a \emph{$D$-dent} is a $D$-edge
whose both endpoints are reflex vertices of the polygon.
We call a sequence of alternating N- and E-edges
a {\em{\rm NE}-staircase\/}, and similarly we define
the {\em{\rm NW}-staircase}, {\em{\rm SE}-staircase}, and
{\em{\rm SW}-staircase}; clearly, each of these staircases is both $x$- and $y$-monotone.
Additionally, we characterize the reflex vertices of an orthogonal polygon based on the type of incident edges; more specifically, each reflex vertex incident to a N-edge and an E-edge is called a {\em{\rm NE}-reflex vertex\/}, and analogously we have the {\em{\rm NW}-, {\rm SE}-} and  {\em{\rm
SW}-reflex vertices}. See Figure~\ref{orthog_basic}, left.
The definition of reflex maxima/minima with respect to some orientation (Definition~\ref{defi2}) and the angles of the lines~$L$ such that both neighbors of a reflex vertex are both below or both above~$L$ imply the following observation.

\begin{observation} \label{obs:NW_SE_0_kernel} \rm
\begin{itemize}
\item [(i)]
For $\theta = 0$ (resp. $\theta = -\frac{\pi}{2}$), only the S- and N-dents (resp. W- and E-dents) contribute reflex minima and maxima, respectively.
\item [(ii)]
With respect to an orientation~$\theta \in (0, \frac{\pi}{2})$,
\emph{every} SE-reflex vertex of an orthogonal polygon is
a reflex maximum
and \emph{every} NW-reflex vertex is a reflex minimum,
whereas for $\theta \in (\frac{\pi}{2}, \pi)$,
\emph{every} SW-reflex vertex is a reflex maximum
and \emph{every} NE-reflex vertex is a reflex minimum.
\\
Analogously,
with respect to the orientation~$\theta+90^\circ$,
for $\theta \in (0, \frac{\pi}{2})$,
\emph{every} SW-reflex vertex is a reflex maximum
and \emph{every} NE-reflex vertex is a reflex minimum,
whereas for $\theta \in (\frac{\pi}{2}, \pi)$,
\emph{every} SE-reflex vertex is a reflex maximum
and \emph{every} NW-reflex vertex is a reflex minimum.
\end{itemize}
\end{observation}

\noindent
As not all SE-reflex and NW-reflex vertices are corners of dents, Observation~\ref{obs:NW_SE_0_kernel} implies that there may be a discontinuity in the  area or perimeter of the $\{0^\circ\}$-$\Kl_{\theta}(P)$ at $\theta = 0$ and $\theta = \frac{\pi}{2}$; these two cases need to be treated separately. Furthermore, it points out a crucial advantage of the orthogonal polygons over simple polygons stated in the following observation.

\begin{observation} \label{obs:ortho_advantage} \rm
In an orthogonal polygon~$P$, for any $\theta \in (0, \frac{\pi}{2})$ (and similarly for any $\theta \in (-\frac{\pi}{2}, 0)$), the~set of reflex minima/maxima {does not change}, and thus the lines bounding the strip~$S_{\theta}(P)$ rotate in a continuous fashion.
\end{observation}

This directly implies that a situation like the one depicted in Figure~\ref{Figure4} cannot occur.
Finally, statement~(ii) of Observation~\ref{obs:NW_SE_0_kernel} implies the following corollary.

\begin{corollary} \label{cor:max_above_min}
Let $P$ be a simple orthogonal polygon.
If there are a SE-reflex vertex $u = (x_u,y_u)$ and a NW-reflex vertex $v = (x_v,y_v)$ of $P$ such that $x_u \le x_v$ and $y_u \ge y_v$, then the $\{0^\circ\}$-$\Kl_{\theta}(P)$ is empty for each $\theta \in (0, \frac{\pi}{2})$.
Similarly, if there are a SW-reflex vertex $u' = (x_{u'},y_{u'})$ and a NE-reflex vertex $v' = (x_{v'},y_{v'})$ of~$P$ such that $x_{u'} \ge x_{v'}$ and $y_{u'} \ge y_{v'}$, then the $\{0^\circ\}$-$\Kl_{\theta}(P)$ is empty for each $\theta \in (\frac{\pi}{2}, \pi)$.
\end{corollary}
\begin{proof}
To see this, note that for any $u,v$ as in the statement of the corollary, for any $\theta \in (\frac{\pi}{2}, \pi)$, the line through $u$ at angle~$\theta$ is above the line through $v$ at angle~$\theta$; see Figure~\ref{orthog_basic}, right. Then, because $u$ and $v$ contribute a reflex maximum and a reflex minimum respectively (Observation~\ref{obs:NW_SE_0_kernel}(ii)), the strip~$S_{\theta}$ is empty and so is the $\{0^\circ\}$-$\Kl_{\theta}(P)$ by Lemma~\ref{rotatedLemma1}.
A similar argument works for the vertices $u',v'$.
\end{proof}

\begin{figure}[htb]
\centering
\subfloat{
\includegraphics[width=0.36\textwidth]{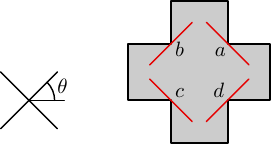}%
}
\qquad\qquad\qquad\qquad\qquad\qquad
\subfloat{
\includegraphics[width=0.2\textwidth]{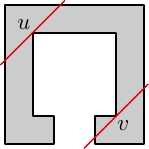}%
}
\caption{Left:~The NW-reflex vertex~$b$ and the SE-reflex vertex~$d$ are a reflex minimum and a reflex maximum with respect to the orientation~$\theta$, %$0^{\circ}_{\theta}$,
respectively, whereas the NE-reflex vertex~$a$ and the SW-reflex vertex~$c$ are a reflex minimum and a reflex maximum with respect to the orientation~${\theta}+\frac{\pi}{2}$, respectively. Right:~Illustration for
Corollary~\ref{cor:max_above_min}.}
\label{orthog_basic}
\end{figure}

\medskip\noindent
\textbf{Notation}. \ We denote by $\vartheta_P(a,b)$ the counterclockwise (CCW) boundary chain of polygon~$P$ from point~$a$ to point~$b$ where $a$ and $b$ are located on the boundary of $P$.

\subsection{The $\{0^\circ\}$-$\Kl_{\theta}(P)$ of simple orthogonal polygons}
\label{subsec:0_orthogonal_case}

We now prove the results in the first row of Table~\ref{Table:ResultsOrthogonal}, focusing on the case for $\theta \in [0, \frac{\pi}{2})$ since the case for $\theta \in [-\frac{\pi}{2}, 0)$ is similar.
Observation~\ref{obs:NW_SE_0_kernel} implies that for $\theta = 0$, the $\{0^\circ\}$-$\Kl_{\theta}(P)$, if non-empty, is determined by a lowest N-dent and a highest S-dent and that for $\theta \in (0, \frac{\pi}{2})$, only the SE-reflex (NW-reflex respectively) vertices contribute reflex maxima (resp. minima).

Let $P$ be a simple orthogonal polygon and suppose that there is at least one SE-reflex vertex in $P$. Let $u$ be the leftmost SE-reflex vertex of~$P$ (in case of ties, take the topmost such vertex), consider the downward-pointing ray~$\vec{r}$ emanating from $u$, and, among its intersections with S-edges of $P$ extending to the left of $\vec{r}$, let $s_{SE}$ be the closest one to $u$.
Similarly, let $u'$ be the topmost SE-reflex vertex of $P$ (in case of ties, take the leftmost such vertex) and let $t_{SE}$ be, among the points of intersection of the rightward-pointing~$\vec{r}\,'$ emanating from $u'$ with an E-edge extending above~$\vec{r}\,'$, the one closest to $u'$; see Figure~\ref{fig_0_chains}, left.

Next, let $C_{SE}$ be the upper hull of the CCW boundary chain~$\vartheta_P(s_{SE},t_{SE})$; the chain~$C_{SE}$ is the blue dashed line in Figure~\ref{fig_0_chains}, left. Similarly, by working with the NW-reflex vertices, we locate the (in case of ties, topmost) rightmost and the (in case of ties, leftmost) bottom-most NW-reflex vertices and we define the points $s_{NW}$ and $t_{NW}$, and the lower hull~$C_{NW}$ of the CCW boundary chain~$\vartheta_P(s_{NW},t_{NW})$.
The definition of the chain~$C_{SE}$ which states that $C_{SE}$ is the upper hull of $\vartheta_P(s_{SE},t_{SE})$ and implies that all the vertices of $C_{SE}$ except for $s_{SE},t_{SE}$ are SE-reflex vertices and the corresponding arguments for the chain~$C_{NW}$ imply the following lemma.

\begin{lemma} \label{lemma:chains}
Let $s_{SE}$, $t_{SE}$, $C_{SE}$, $s_{NW}$, $t_{NW}$, and $C_{NW}$ of a simple orthogonal polygon~$P$ be as defined earlier. If all SE-reflex vertices belong to the CCW boundary chain~$\vartheta_P(s_{SE},t_{SE})$ and all NW-reflex vertices belong to the CCW boundary chain~$\vartheta_P(s_{NW},t_{NW})$, then for any angle~$\theta \in (0, \frac{\pi}{2})$, any  vertex of $C_{SE}$ (resp. $C_{NW}$) at which a line at angle~$\theta$ is tangent to $C_{SE}$ (resp. $C_{NW}$) is a topmost reflex maximum (resp. lowest reflex minimum) with respect to the orientation at angle~$\theta$.
\end{lemma}

\begin{figure}[htb]
\centering
\subfloat{
\includegraphics[width=0.4\textwidth]{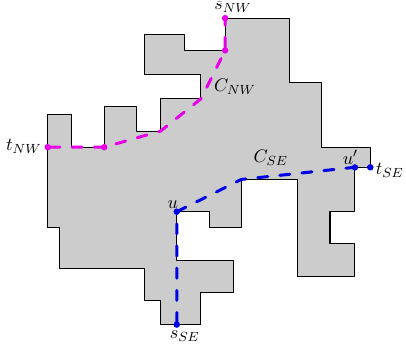}%
}
\hspace{1.5cm}
\subfloat{
\includegraphics[width=0.45\textwidth]{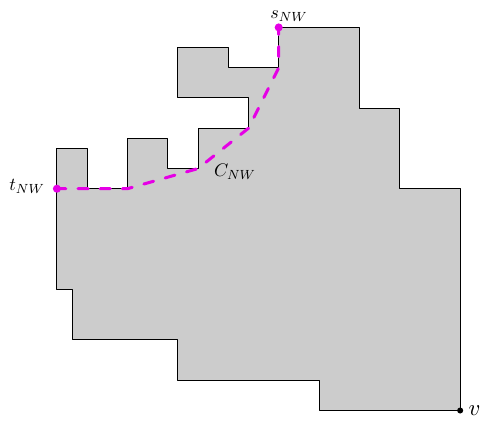}%
}
\caption{Left:~An orthogonal polygon and the corresponding convex chains $C_{SE}$ and $C_{NW}$. Right:~An orthogonal polygon without SE-reflex vertices in which we can consider that the convex chain~$C_{SE}$ degenerates into vertex~$v$.}
\label{fig_0_chains}
\end{figure}

Additionally, assuming that the CCW ordering of $s_{SE}$, $t_{SE}$, $s_{NW}$, and $t_{NW}$ around the boundary of~$P$ is precisely $s_{SE}, t_{SE}, s_{NW}, t_{NW}$, we can prove the following property of the CCW boundary chains of $P$ from $t_{NW}$ to $s_{SE}$ and from $t_{SE}$ to $s_{NW}$.

\begin{lemma} \label{lemma:staircases}
Let $s_{SE}$, $t_{SE}$, $s_{NW}$, and $t_{NW}$ of a simple polygon~$P$ be as defined earlier, and assume that the CCW ordering of $s_{SE}$, $t_{SE}$, $s_{NW}$, and $t_{NW}$ around the boundary of $P$ is precisely $s_{SE}, t_{SE}, s_{NW}, t_{NW}$ and that all SE-reflex vertices belong to the CCW boundary chain~$\vartheta_P(s_{SE},t_{SE})$ and all NW-reflex vertices belong to the CCW boundary chain~$\vartheta_P(s_{NW},t_{NW})$.
Then, the CCW boundary chain~$\vartheta_P(t_{NW},s_{SE})$ of $P$ from $t_{NW}$ to $s_{SE}$ is a SW-staircase and the CCW boundary chain~$\vartheta_P(t_{SE},s_{NW})$ from $t_{SE}$ to $s_{NW}$ is a NE-staircase.
\end{lemma}

\begin{proof}
Let us consider the case of the CCW boundary chain~$\vartheta_P(t_{NW},s_{SE})$ (see Figure~\ref{fig_0_chains}, left); the proof for the chain~$\vartheta_P(t_{SE},s_{NW})$ is symmetric. Since all SE-reflex vertices belong to the CCW boundary chain~$\vartheta_P(s_{SE},t_{SE})$ and all NW-reflex vertices belong to the CCW boundary chain~$\vartheta_P(s_{NW},t_{NW})$, the chain~$\vartheta_P(t_{NW},s_{SE})$ contains neither SE-reflex nor NW-reflex vertices.

Suppose that we start at the W-edge to which $t_{NW}$ belongs (let this edge be $u v$ with $v$ below $u$) and proceed in CCW order. The edge following the W-edge~$u v$ is not a N-edge, otherwise the vertex~$v$ would be a NW-reflex vertex, a contradiction. Thus, the edge following the W-edge~$u v$ is a S-edge, let it be $v w$. If $s_{SE} \in v w$, then we are done and the lemma holds. Otherwise, if the edge following the edge~$v w$ was an E-edge, then the top vertex of the leftmost edge in the CCW boundary chain~$\vartheta_P(w,s_{SE})$ would be a SE-reflex vertex (note that the E-edge incident on $w$ belongs to this chain), a contradiction. Therefore, the edge following the S-edge~$v w$ is a W-edge. Then, the above argument can be repeated until we reach the point~$s_{SE}$, implying that the CCW boundary chain~$\vartheta_P(t_{NW},s_{SE})$ is a NW-staircase.
\end{proof}

Lemma~\ref{lemma:staircases} implies that if the given polygon~$P$ has no SE-reflex vertices, then the CCW boundary chain~$\vartheta_P(t_{NW},s_{NW})$ consists of a SW-staircase followed by a NE-staircase; see Figure~\ref{fig_0_chains}, right.
A similar result holds if there are no NW-reflex vertices.

\subsubsection{The existence of the $\{0^\circ\}$-$\Kl_{\theta}(P)$ for a simple orthogonal polygon~$P$}
\label{sec:0-kernel-ortho-polygons}

In this subsection, we give an algorithm to determine when the $\{0^\circ\}$-$\Kl_{\theta}(P)$ for a simple orthogonal polygon~$P$ is non-empty. First, if no SE-reflex vertex exists, then no S-dents exist and as mentioned, the chain~$C_{SE}$ degenerates into the rightmost lowest vertex (see Figure~\ref{fig_0_chains}, right) which thus belongs to the $\{0^\circ\}$-$\Kl_{\theta}(P)$ for all $\theta \in (0, \frac{\pi}{2})$; thus, the $\{0^\circ\}$-$\Kl_{\theta}(P)$ is non-empty for all $\theta \in [0, \frac{\pi}{2})$. A similar argument holds if no NW-reflex vertex exists. So, in the following, we assume that the polygon~$P$ has SE-reflex and NW-reflex vertices. Then, we show the following lemma.

\begin{lemma} \label{lemma:empty_orthog_0_kernel}
Let $s_{SE}$, $t_{SE}$, $C_{SE}$, $s_{NW}$, $t_{NW}$, and $C_{NW}$ of a simple orthogonal polygon~$P$ be as defined earlier.
\begin{itemize}
\item[(i)]
Let $Q_{SE}$ be the convex part of the plane bounded from the left and above by $C_{SE}$, the downward-pointing ray emanating from $s_{SE}$, and the rightward-pointing ray emanating from $t_{SE}$. Similarly, let $Q_{NW}$ be the convex part of the plane bounded from the right and below by $C_{NW}$, the upward-pointing ray emanating from $s_{NW}$, and the leftward-pointing ray emanating from $t_{NW}$.
\begin{itemize}
  \item[{\em (a)}]
  If the interiors of $Q_{SE}$ and $Q_{NW}$ intersect, then the $\{0^{\circ}\}$-$\Kl_{\theta}(P)$ is empty for each $\theta \in (0, \frac{\pi}{2})$.
  \item[{\em (b)}]
  If the interiors of $Q_{SE}$ and $Q_{NW}$ do not intersect but $Q_{SE}$ and $Q_{NW}$ touch at a common point~$z$, then the $\{0^{\circ}\}$-$\Kl_{\theta}(P)$ degenerates to a line segment for each $\theta$ equal to the angle
  of each common interior tangent of $C_{SE}$ and $C_{NW}$ at $z$, and is empty for all other values of $\theta$.
\end{itemize}
\item[(ii)]
If there exists a SE-reflex vertex not belonging to the CCW boundary chain~$\vartheta_P(s_{SE},t_{SE})$ or a NW-reflex vertex not belonging to the CCW boundary
chain~$\vartheta_P(s_{SE},t_{SE})$,
then the $\{0^{\circ}\}$-$\Kl_{\theta}(P)$ is empty for each $\theta \in (0, \frac{\pi}{2})$.
\end{itemize}
\end{lemma}

\begin{proof}
(i.a)~Let $p$ be a point in the intersection of the interiors of the unbounded convex polygons $Q_{SE}$ and~$Q_{NW}$. Then, for any angle $\theta \in [0, \frac{\pi}{2})$,
$p$ lies below the tangent to~$C_{SE}$ at angle $\theta$ and above the tangent to $C_{NW}$ at angle $\theta$ and thus the strip~$S_{\theta}(P)$ is empty. Therefore, for each $\theta \in [0, \frac{\pi}{2})$, the strip~$S_{\theta}$ is empty and, by Lemma~\ref{rotatedLemma1}, so is the $\{0^\circ\}$-$\Kl_{\theta}(P)$.

\begin{figure}[htb]
\centering

\includegraphics[width=0.5
\textwidth]{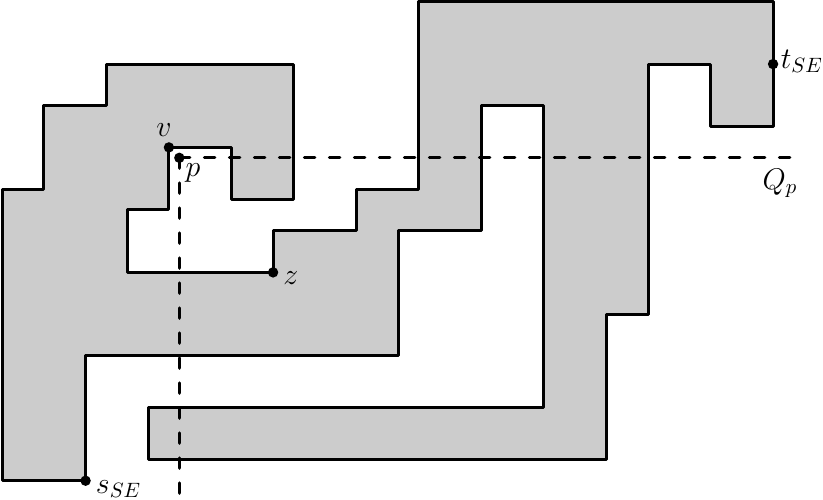}

\caption{The quadrant~$Q_p$ in the proof of Lemma~\ref{lemma:empty_orthog_0_kernel}(ii).}
\label{fig_empty_orthog_O_kernel}
\end{figure}

\noindent
(i.b)~If $Q_{SE}$ and $Q_{NW}$ touch along their horizontal rays, then the $\{0^\circ\}$-$\Kl_{\theta}(P)$ is a horizontal line segment if $\theta = 0$, otherwise it is empty. Similarly, if they touch along their vertical rays, then the $\{0^\circ\}$-$\Kl_{\theta}(P)$ is a vertical line segment if $\theta = \frac{\pi}{2}$, otherwise it is empty.
Next, assume that $Q_{SE}$ and $Q_{NW}$ touch at a point of $C_{SE}$ and $C_{NW}$. Then, because $Q_{SE}$ and $Q_{NW}$ are convex, they touch at a connected portion of $C_{SE}$ and $C_{NW}$, that is, they touch at a point or a line segment.
In either case, for any angle~$\theta$ of
any common interior tangent to $C_{SE}$ and $C_{NW}$, the $\{0^\circ\}$-$\Kl_{\theta}(P)$ is a line segment, otherwise it is empty.

\medskip\noindent
(ii)~Let us concentrate on the case of a SE-reflex vertex of $P$, say $v$, not belonging to the CCW boundary chain $\theta_P(s_{SE},t_{SE})$.
(The case of a NW-reflex vertex not belonging to the CCW boundary chain~$\vartheta_P(s_{NW},t_{NW})$ is similar.)
Let $p$ be a point infinitesimally to the right and below $v$ so that $p$ is outside $P$. Since the chain~$\vartheta_P(s_{SE},t_{SE})$ is determined by the leftmost and the topmost SE-reflex vertices, the $x$-coordinate of~$v$ is larger than the $x$-coordinate of $s_{SE}$ while
the $y$-coordinate of~$v$ is smaller than the $y$-coordinate of $t_{SE}$. This implies that $\vartheta_P(s_{SE},t_{SE})$ intersects both the rightward-pointing horizontal ray~$\vec{r}_\rightarrow$ emanating from $p$ and the downward-pointing vertical ray~$\vec{r}_\downarrow$ emanating from $p$. Let $Q_p$ be the closed quadrant delimited by the rays $\vec{r}_\rightarrow$ and $\vec{r}_\downarrow$ (see Figure~\ref{fig_empty_orthog_O_kernel}). Consider the set~$A_C$ of all minimal boundary chains of $P$ that lie in $Q_P$ and are delimited by a point on $\vec{r}_\rightarrow$ and a point on $\vec{r}_\downarrow$ (the minimality implies that no point in such a chain other than its endpoints belongs to either $\vec{r}_\rightarrow$ or $\vec{r}_\downarrow$); these chains do not intersect, therefore they are totally ordered and the ordering is the same as the ordering of their endpoints on $\vec{r}_\rightarrow$ or $\vec{r}_\downarrow$. Among the chains in $A_C$, let $C$ be the chain with endpoint on $\vec{r}_\rightarrow$ closest to $p$. Since $p$ is outside the polygon~$P$, then the interior of $P$ is to the left of $C$ as we walk along it from its endpoint on $\vec{r}_\rightarrow$ to its endpoint on $\vec{r}_\downarrow$; see Figure~\ref{fig_empty_orthog_O_kernel}. Then, the right vertex of a lowest horizontal edge in $C$ is a NW-reflex vertex (e.g., vertex~$z$ in Figure~\ref{fig_empty_orthog_O_kernel}). Since this vertex is below and to the right of the SE-reflex vertex~$v$, Corollary~\ref{cor:max_above_min} implies that the $\{0^\circ\}$-$\Kl_{\theta}(P)$ is empty for each $\theta \in (0, \frac{\pi}{2})$.
\end{proof}

So, assume that none of the cases of Lemma~\ref{lemma:empty_orthog_0_kernel} holds. Then, the chains $C_{SE}$ and $C_{NW}$ do neither intersect nor touch, and the inner common tangents to $C_{SE}$ and $C_{NW}$ are well defined; let them be $T_1$ and $T_2$ with the slope of $T_1$ being smaller than the slope of $T_2$, and let $\phi_1, \phi_2$ be the CCW angle with respect to the positive $x$-axis of $T_1$ and $T_2$, respectively. If the $y$-coordinate of $t_{NW}$ is greater than the $y$-coordinate of $t_{SE}$, then we set $\theta_{min} = 0$, otherwise $\theta_{min} = \phi_1$. Similarly, we define $\theta_{max}$ to be equal to~$\frac{\pi}{2}$ if the $x$-coordinate of $s_{SE}$ is greater than the $x$-coordinate of $s_{NW}$, otherwise $\theta_{max} = \phi_2$. For example, in Figure~\ref{fig_0_chains}, left, $\theta_{min} = 0$ and $\theta_{max} < \frac{\pi}{2}$.
Then, since for $\theta \in (0, \frac{\pi}{2})$ the strip~$S_{\theta}(P)$ is {non-empty} if and only if $\theta \in [\theta_{min}, \theta_{max}] \cap (0,\frac{\pi}{2})$, by Lemma~\ref{lema1} we have:

\begin{lemma} \label{lemma:range_orthog_0_kernel}
Let $P$ be a simple orthogonal polygon such that none of the cases of Lemma~\ref{lemma:empty_orthog_0_kernel} hold and consider that $\theta \in (0, \frac{\pi}{2})$. Then, if no SE-reflex or NW-reflex vertices exist, the $\{0^{\circ}\}$-$\Kl_{\theta}(P)$ is {non-empty} for each $\theta \in (0, \frac{\pi}{2})$, otherwise the $\{0^{\circ}\}$-$\Kl_{\theta}(P)$ is {non-empty} if and only if $\theta \in [\theta_{min}, \theta_{max}] \cap (0,\frac{\pi}{2})$.
\end{lemma}

\begin{corollary}\label{corol:orthog_0_angle_size}
The values of the angle~$\theta \in [0, \frac{\pi}{2})$ for which the $\{0^\circ\}$-$\Kl_{\theta}(P)$ of a simple orthogonal polygon~$P$ is non-empty form a single interval and potentially the value $\theta = 0$.
\end{corollary}

Based on the above discussion, we outline our algorithm in Algorithm~\ref{ortho_algo1}.

\begin{algorithm}[ht!]
\caption{\label{ortho_algo1}Computing the intervals of $\theta$ such that $\{0^\circ\}$-$\Kl_{\theta}(P) \ne \emptyset$ for a simple orthogonal polygon~$P$}
\ \\ \
\hglue 5pt \textbf{Input:} A simple orthogonal polygon $P$ with $n$ vertices\\
\hglue 5pt \textbf{Output:} The intervals of the angle~$\theta$ such that $\{0^\circ\}$-$\Kl_{\theta}(P) \ne \emptyset$
\begin{algorithmic}[1]
    \Statex
    \Statex \hglue -6mm {\sc{\bf STEP 1:}  Check if $\{0^\circ\}$-$\Kl_{\theta}(P) \ne \emptyset$ for $\theta = 0$}
    \If {no SE-reflex vertices exist or no NW-reflex vertices exist}
    \State \textbf{output} $[0,\frac{\pi}{2})$ and \textbf{stop}
    \EndIf
    \State compute the N- and S-dents of $P$ and let $e_N$ and $e_S$ be a lowest N-dent and a highest S-dent, respectively
    \If {the $y$-coordinate of $e_N$ is smaller than the $y$-coordinate of $e_S$}
    \State $solution\_for\_0 \gets \emptyset$
    \Else
    \State $solution\_for\_0 \gets [0,0]$
    \EndIf

    \Statex
    \Statex \hglue -6mm {\sc{\bf STEP 2:}  Check if $\{0^\circ\}$-$\Kl_{\theta}(P) \ne \emptyset$ for $\theta \in (0,\frac{\pi}{2})$}
    \State compute the points $s_{SE}, t_{SE}, s_{NW}, t_{NW}$ and the convex chains $C_{SE}$ and~$C_{NW}$
    \State check whether $C_{SE}$ and $C_{NW}$ cross, touch, or do not intersect
    \If {there exists a SE-reflex vertex not in $\vartheta_P(s_{SE},t_{SE})$ \,or\, a NW-reflex vertex not in $\vartheta_P(s_{NW},t_{NW})$ \,or\, the chains $C_{SE}$ and $C_{NW}$ cross}
    \State $solution\_for\_0$-$\frac{\pi}{2} \gets \emptyset$
    \ElsIf {the chains $C_{SE}$ and $C_{NW}$ share a line segment~$I$}
      \State $solution\_for\_0$-$\frac{\pi}{2} \gets [\theta_I,\theta_I] \cap (0,\frac{\pi}{2})$ where $\theta_I$ is the angle of the line supporting the line segment~$I$
      \ElsIf {the chains $C_{SE}$ and $C_{NW}$ touch at a single point~$z$}
        \State let $R_{SE}(z)$ ($R_{NW}(z)$, resp.) be the angle interval of the tangent to $C_{SE}$ (to $C_{NW}$ resp.) at $z$
        \State $solution\_for\_0$-$\frac{\pi}{2} \gets R_{SE}(z) \cap R_{NW}(z) \cap (0,\frac{\pi}{2})$
        \Else
          \State compute the inner common tangents to $C_{SE}$ and $C_{NW}$
          \State compute angles $\theta_{min}$ and $\theta_{max}$ as explained in the paragraph preceding Lemma~\ref{lemma:range_orthog_0_kernel}
          \State $solution\_for\_0$-$\frac{\pi}{2} \gets [\theta_{min}, \theta_{max}] \cap (0,\frac{\pi}{2})$
    \EndIf

    \Statex
    \Statex \hglue -6mm {\sc{\bf STEP 3:}  Output results}
    \State \textbf{output} $solution\_for\_0 \cup solution\_for\_0$-$\frac{\pi}{2}$
\end{algorithmic}
\end{algorithm}

\noindent\emph{Analysis of Algorithm~\ref{ortho_algo1}}. The correctness of Algorithm~\ref{ortho_algo1} follows from the fact that if no SE-reflex or no NW-reflex vertices exist, the $\{0^\circ\}$-$\Kl_{\theta}(P)$ is non-empty for all $\theta \in [0, \frac{\pi}{2})$, and from Observation~\ref{obs:NW_SE_0_kernel} and Lemmas \ref{lemma:empty_orthog_0_kernel} and \ref{lemma:range_orthog_0_kernel}.

Computing the SE-reflex and NW-reflex vertices,  the N- and S-dents, and then finding a lowest N-dent and a highest S-dent can be done in $O(n)$ time. Thus, STEP~1 can be completed in $O(n)$ time and $O(1)$ space.
Computing the points $s_{SE}, t_{SE}, s_{NW}, t_{NW}$ can be done in $O(n)$ time. The chains $C_{SE}$ and $C_{NW}$ can be computed in $O(n)$ time as well \cite{M1987}. As the size of $C_{SE}$ and $C_{NW}$ is $O(n)$ and they are $x$-monotone, we can check whether they cross or touch in $O(n)$ time by walking along them from their leftmost to their rightmost endpoint in lockstep fashion. Computing the angle of the line supporting the segment~$I$ and the angle ranges of the tangents at~$z$ can be done in $O(1)$ time. The inner common tangents to $C_{SE}$ and $C_{NW}$ can be computed in $O(\log n)$ time (in a fashion similar to computing the outer ones \cite{KS}), from which we can compute $\theta_{min}$ and $\theta_{max}$ in~$O(1)$ time. Hence, STEP~$2$ requires $O(n)$ time and $O(n)$ space. Finally, STEP~$3$ takes $O(1)$ time and space. In summary, we have:

\begin{theorem}\label{thm:0-kernel-exists-ortho}
For a simple orthogonal polygon~$P$ with $n$ vertices, the intervals of $\theta \in [-\frac{\pi}{2},\frac{\pi}{2})$ for which $\{0^\circ\}$-$\Kl_{\theta}(P) \ne \emptyset$ can be computed in $O(n)$ time and space.
\end{theorem}

\subsubsection{Optimizing the area and perimeter of the $\{0^\circ\}$-$\Kl_{\theta}(P)$ for a simple orthogonal polygon~$P$}
\label{subsec:opt_area_orthog_0-kernel}

In this subsection, we present an algorithm that computes an angle~$\theta \in [0, \frac{\pi}{2})$ such that the area (or perimeter)
of the $\{0^{\circ}\}$-$\Kl_{\theta}(P)$
is maximized; minimization works similarly.

If the $\{0^\circ\}$-$\Kl_{\theta}(P)$ for $\theta = 0$ is non-empty, we compute its area/perimeter and we use these to set the current maximum value and the current angle of the maximum that we maintain; if the $\{0^\circ\}$-$\Kl_{0}(P)$ is empty, then its area and perimeter are set to $0$. Next, we work for $\theta \in (0, \frac{\pi}{2})$. For the sake of generality, in the following, we consider that the polygon~$P$ has both SE-reflex and NW-reflex vertices; if one of these two vertex types is missing, then we skip the computations involving that vertex type, whereas if both vertex types are missing, then the kernel is the entire polygon~$P$ and we simply need to compute the area or perimeter of $P$.

\begin{figure}[htb]
\centering
\vspace{-1mm}
\subfloat{
\includegraphics[width=0.37\textwidth]{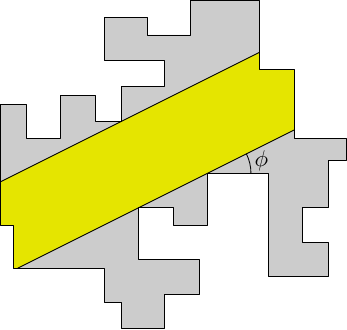}%
}
\hspace{1.5cm}
\subfloat{
\includegraphics[width=0.37\textwidth]{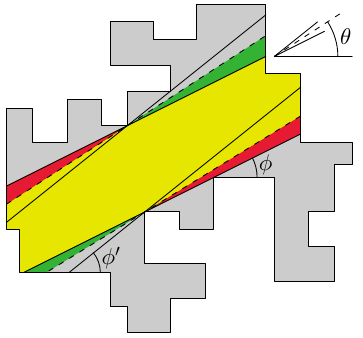}%
}
\vspace{-1.5mm}
\caption{Left:~$\{0^\circ\}$-$\Kl_{\theta}(P)$ for $\theta=\phi$. Right:~Optimizing the area/perimeter of the $\{0^\circ\}$-$\Kl_{\theta}(P)$ for~$\theta \in [\phi, \phi')$.}
\label{fig_1}
\end{figure}

Next, we check whether the conditions of Lemma~\ref{lemma:empty_orthog_0_kernel}
hold; if they do, the area of each of the degenerate kernels that arise is equal to $0$, whereas, whenever the kernel is non-empty, its perimeter can be computed in $O(1)$ time. Otherwise, we compute the interval $A = [\theta_{min},\theta_{max}] \cap [0,\frac{\pi}{2})$ as in Algorithm~\ref{ortho_algo1}; we need to maximize the area or perimeter of the $\{0^{\circ}\}$-$\Kl_{\theta}(P)$ for any angle~$\theta \in A$. We start at $\theta = \theta_{min}$ and we explicitly compute
the $\{0^{\circ}\}$-$\Kl_{\theta_{min}}(P)$ and its area (perimeter), which is the current area (perimeter) maximum (we note that if $\theta_{min} = 0$, we compute the area of the intersection of the polygon~$P$ with a horizontal strip defined by the highest SE-reflex vertex from below and the lowest NW-reflex vertex from above).
Subsequently, as in Section~\ref{sec:0-kernel-simple-polygons}, we partition the interval $A$ into angular subintervals, in each of which the following property holds:

\medskip\noindent
\textbf{Property 1}
The kernel involves the same topmost reflex maximum and lowest reflex minimum and the same edges of the polygon.

\medskip\noindent
For the resulting partition, say $P_A$, the following lemma holds.

\begin{lemma} \label{lemma:subinterval_number}
For an orthogonal polygon~$P$ with $n$ vertices, the size of the partition~$P_a$ of $A = [\theta_{min},\theta_{max}] \cap [0,\frac{\pi}{2})$ is $O(n)$.
\end{lemma}

\begin{proof}
Let $P_{SE}$ be the partition of the interval $A = [\theta_{min},\theta_{max}] \cap [0,\frac{\pi}{2})$ based on which vertex of the chain~$C_{SE}$ is the current topmost reflex maximum and on which edges of the polygon bound the lower segment of the strip~$S_{\theta}$. Then, an angle~$\theta \in A$ is a partition point if it is
\begin{itemize}
\item the angle of an edge of the chain~$C_{SE}$ (see Lemma~\ref{lemma:chains} because at that angle the topmost reflex maximum changes, or
\item the angle of the tangent from a vertex of $P$ to the chain~$C_{SE}$ because at that point the lower segment of the strip~$S_{\theta}$ moves to another edge.
\end{itemize}
Because the segments bounding the strip~$S_{\theta}$ rotate in a continuous fashion (Observation~\ref{obs:ortho_advantage}), the number of vertices of the chain~$C_{SE}$ and the polygon is $O(n)$, the size of $P_{SE}$ is $O(n)$. Similarly, the size of the corresponding partition~$P_{NW}$ related to the current lowest reflex minimum and the chain~$C_{NW}$ is $O(n)$ as well. Then, the partition~$P_A$ is the refinement of the partition~$P_{SE}$ by means of the partition~$P_{NW}$, which yields that its size is $O(n)$.
\end{proof}

After the partition~$P_A$ has been computed, we process the subintervals in increasing angle value and in each such interval $[\beta_j, \beta_{j+1})$, we maximize the area/perimeter as a function of an angle $\beta \in [\beta_j, \beta_{j+1})$ by taking into account the area/perimeter of $\{0^{\circ}\}$-$\Kl_{\beta_j}(P)$ and of the two green triangles and the two red triangles in the spirit of Equation \ref{eq:area}, as shown in Figure~\ref{fig_1}, right. The area (respectively perimeter) of each of these four triangles depends linearly on $\tan\beta$ and $\cot\beta$ (resp{.} linearly on $(1\pm\cos\beta)/\sin\beta$ and $(1\pm\sin\beta)/\cos\beta$), see the appendix.

Based on the above discussion, we outline our algorithm to maximize the area of $\{0^\circ\}$-$\Kl_{\theta}(P)$ in Algorithm~\ref{ortho_algo2}.

\begin{algorithm}[ht!]
\caption{\label{ortho_algo2}Computing the maximum area of $\{0^\circ\}$-$\Kl_{\theta}(P)$ for a simple orthogonal polygon~$P$
}
\ \\ \
\hglue 5pt \textbf{Input:} A simple orthogonal polygon $P$ with $n$ vertices\\
\hglue 5pt \textbf{Output:} A value of the angle~$\theta$ such that the area of $\{0^\circ\}$-$\Kl_{\theta}(P)$ is maximum
\begin{algorithmic}[1]
    \Statex
    \Statex \hglue -6mm {\sc{\bf STEP 1:}  Check special cases}
    \State execute Algorithm~\ref{ortho_algo1} to compute the set~$T$ of values of $\theta$ for which $\{0^\circ\}$-$\Kl_{\theta}(P) \ne \emptyset$, and $\theta_{min}, \theta_{max}$, if they can be defined
    \State $current\_angle \gets 0$
    \If {$0 \in T$}
    \State compute the area of $\{0^\circ\}$-$\Kl_{\theta}(P)$ for $\theta = 0$
    \State $current\_max \gets$ computed area
    \Else
    \State $current\_max \gets 0$
    \EndIf
    \If {any of the conditions of Lemma~\ref{lemma:empty_orthog_0_kernel} holds}
    \State \textbf{output} $current\_max$ \quad and \textbf{stop}
    \EndIf

    \Statex
    \Statex \hglue -6mm {\sc{\bf STEP 2:}  Maximize the area of $\{0^\circ\}$-$\Kl_{\theta}(P)$ for $\theta \in A = [\theta_{min},\theta_{max}] \cap [0,\frac{\pi}{2})$}
    \State compute the area of $\{0^\circ\}$-$\Kl_{\theta}(P)$ for $\theta = \theta_{min}$
    \State $\theta \gets \theta_{min}$
    \While {$\theta < \theta_{max}$}
    \State compute the angles for which the highest reflex maximum and the lowest reflex minimum change
    \State compute the angles for each of the segments bounding the strip~$S_{\theta}$ to reach the next vertex of the polygon~$P$
    \State $\delta \gets$ the minimum among the angles computed in the $2$ preceding lines
    \State maximize the area of the $\{0^\circ\}$-$\Kl_{\theta}(P)$ for $\theta \in [\theta, \theta+\delta)$ by using the expressions for the area in the appendix
    \State update, if needed, the current maximum area value $current\_max$ and the corresponding angle $current\_angle$
    \EndWhile
    \State \textbf{output} $current\_max$ and $current\_angle$
\end{algorithmic}
\end{algorithm}

\bigskip
\noindent\emph{Analysis of Algorithm~\ref{ortho_algo2}}. The correctness of Algorithm~\ref{ortho_algo2} follows from Observation~\ref{obs:NW_SE_0_kernel}, Lemmas \ref{lemma:staircases}, \ref{lemma:empty_orthog_0_kernel}, and~\ref{lemma:range_orthog_0_kernel}, and the preceding discussion.
Algorithm~\ref{ortho_algo1} requires $O(n)$ time and space, as do the computation of the $\{0^\circ\}$-$\Kl_{\theta}(P)$ for $\theta = 0$ and its area, and checking the conditions of Lemma~\ref{lemma:empty_orthog_0_kernel}. Thus, STEP~1 takes $O(n)$ time and space.
Computing the $\{0^\circ\}$-$\Kl_{\theta}(P)$ for $\theta = \theta_{min}$ can be explicitly done in $O(n)$ time and space. Each iteration of the \texttt{while} loop in STEP~2 takes $O(1)$ time as it involves accessing and processing in $O(1)$ time at most $8$ neighboring vertices and maximizing a constant-degree polynomial. Moreover, it is important to note that the subintervals processed in the \texttt{while} loop precisely form the partition~$P_A$. Since a different subinterval is processed in each iteration of the while loop and since the number of subintervals is $O(n)$
(Lemma~\ref{lemma:subinterval_number}), the execution of the \texttt{while} loop in STEP~2 takes $O(n)$ time. Hence,
by also taking into account that minimization, where meaningful, can be handled analogously, we get:

\begin{theorem}\label{thm:0-kernel-optareaperim-ortho}
For a simple orthogonal polygon $P$ with $n$ vertices, the values of $\theta$ such that the area or the perimeter of the $\{0^{\circ}\}$-$\Kl_{\theta}(P)$ are maximum/minimum can be computed in $O(n)$ time and space.
\end{theorem}

\subsection{The rotated $\{0^\circ,90^\circ\}$-$\Kl_{\theta}(P)$ of a simple orthogonal polygon~$P$}
\label{subsec:0_90_orthogonal_case}

We now extend our study to $\mathcal{O}=\{0^\circ, 90^\circ \}$ for a simple orthogonal polygon~$P$, proving the results in the second row of Table~\ref{Table:ResultsOrthogonal}. Observe that it suffices to consider $\theta \in [0,\frac{\pi}{2})$ since $\{0^\circ,90^\circ\}$-$\Kl_{0}(P)= \{0^\circ,90^\circ\}$-$\Kl_{\frac{\pi}{2}}(P)$. Again, Observation~\ref{obs:O-kernel} and Lemma~\ref{rotatedLemma1} imply that
\begin{equation}\label{eq:0_90_ortho}
\{0^\circ, 90^\circ\}\hbox{-}\Kl_{\theta}(P) \ =\  \{0^\circ\}\hbox{-}\Kl_{\theta}(P) \ \cap\  \{90^\circ\}\hbox{-}\Kl_{\theta}(P) \ =\  \left(\, S_{\theta} (P) \cap P \,\right) \ \cap\  \left(\, S_{\theta+90^\circ} (P) \cap P \,\right)
\end{equation}
and therefore $\{0^\circ, 90^\circ\}\hbox{-}\Kl_{\theta}(P) = S_{\theta} (P) \cap S_{\theta+90^\circ} (P) \cap P$, that is, the case is an extension of the $\{0^\circ\}$-$\Kl_{\theta}(P)$ with {two} strips $S_{\theta}(P)$ and $S_{\theta+90^\circ}(P)$, which are perpendicular to each other.

For $\theta = 0$, the $\{0^\circ, 90^\circ\}$-$\Kl_{\theta}(P)$ is the intersection of the polygon~$P$ with the horizontal strip determined above by the lowest N-dent and below by the topmost S-dent, and with the vertical strip determined to the left by the rightmost W-dent and to the right by the leftmost E-dent. Thus, the $\{0^\circ, 90^\circ\}$-$\Kl_{0}(P)$ may have reflex vertices (but no dents) at the top left, top right, bottom left or bottom right corners and is orthogonally convex.

Below we consider the case for $\theta \in (0,\frac{\pi}{2})$. In accordance with Observation~\ref{obs:NW_SE_0_kernel}, all reflex vertices are reflex maxima or minima with respect to one of the orientations in $\mathcal{O}_{\theta}$; then, the definition of the $\{0^\circ, 90^\circ\}$-$\Kl_{\theta}(P)$ implies that:

\begin{observation} \label{obs:convex_0_90_kernel} \rm
For $\theta \in (0, \frac{\pi}{2})$,
the $\{0^\circ, 90^\circ\}$-$\Kl_{\theta}(P)$ is a {convex} polygon.
\end{observation}

See Figure~\ref{fig_0_90_chains}, right for an example. Recall that for $\theta = 0$, the $\{0^\circ, 90^\circ\}$-$\Kl_{\theta}(P)$ is not necessarily convex, but it is orthogonally convex.

\begin{figure}[htb]

\centering
\vspace{-1mm}
\subfloat{
\includegraphics[width=0.29\textwidth]{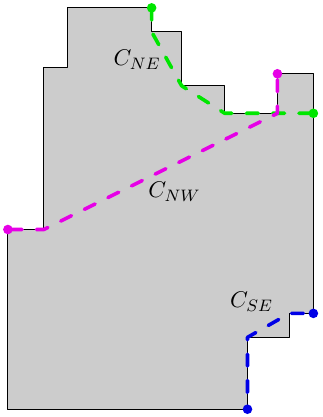}%
}
\hspace{2cm}
\subfloat{
\includegraphics[width=0.29\textwidth]{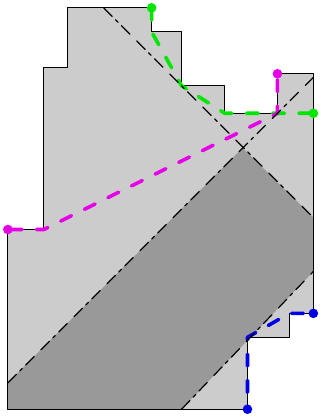}%
}
\vspace{1.5mm}
\caption{Left:~A simple orthogonal polygon~$P$ and the convex chains $C_{SE}, C_{NE}, C_{NW}$; no SW-reflex vertices exist. Right:~The $\{0^\circ,90^\circ\}$-$\Kl_{\theta}(P)$ for $\theta = \frac{\pi}{4}$ is shown darker.}
\label{fig_0_90_chains}
\end{figure}

As mentioned above, in this case, the kernel is in general defined by the two perpendicular strips $S_{\theta}(P)$ and $S_{\theta+90^\circ}(P)$. Let us investigate the cases that may arise for the points of intersection of the lines bounding these strips. So, consider an angle $\theta \in (0, \frac{\pi}{2})$ such that there is at least one reflex maximum in the orientation~$\theta$ and at least one reflex minimum in the orientation~$\theta+90^\circ$, and let $\ell_{{\theta},\downarrow}$ (resp{.} $\ell_{{90^{\circ}+\theta},\uparrow}$) be the bottom (resp{.} top) segment bounding $S_{\theta}(P)$ (resp{.} $S_{\theta+90^\circ}(P)$).
Moreover, let $p$ (resp{.} $q$) be the right endpoint of $\ell_{{\theta},\downarrow}$ (resp{.} $\ell_{{90^{\circ}+\theta},\uparrow}$).
Clearly $p$ belongs to
a N- or an E-edge, and similarly, $q$ belongs to
a S- or an E-edge. Each of the above possibilities for $p$ and $q$ may well arise if the segments $\ell_{{\theta},\downarrow}$ and $\ell_{{90^{\circ}+\theta},\uparrow}$ intersect (see Figure~\ref{fig_0_90_corners}, left); the point of intersection lies in the polygon~$P$ and, as the strips rotate, it moves along an arc of a circle whose diameter is the line segment connecting the reflex maximum and the reflex minimum about which $\ell_{{\theta},\downarrow}$ and $\ell_{{90^{\circ}+\theta},\uparrow}$, respectively, rotate.
However, if these two segments do not intersect, then only one case for the relative location of $p$ and $q$ is possible, as we show in the following lemma.

\begin{figure}[tb]

\centering
\vspace{-1mm}
\subfloat{
\includegraphics[width=0.20\textwidth]{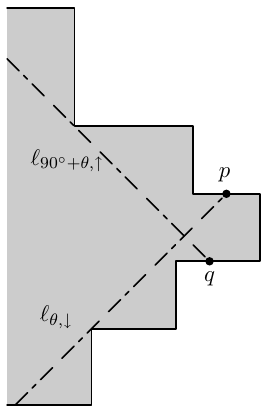}%
}
\hspace{1.5cm}
\subfloat{
\includegraphics[width=0.20\textwidth]{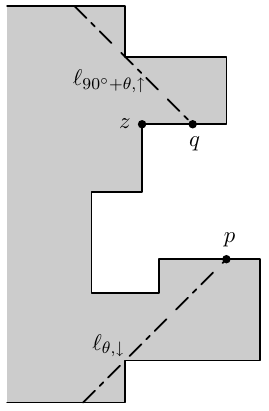}%
}
\hspace{1.5cm}
\subfloat{
\includegraphics[width=0.20\textwidth]{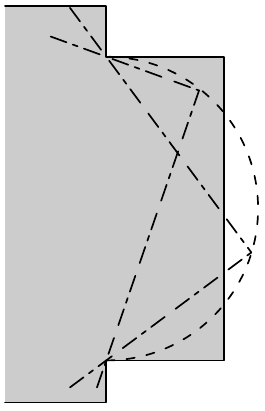}%
}\vspace{-1.5mm}
\caption{Left:~The segments $\ell_{{\theta},\downarrow}$ and $\ell_{{90^{\circ}+\theta},\uparrow}$ intersect.
Middle:~For Lemma~\ref{lemma:non-intersect_segm}; an impossible configuration.
Right:~As the angle~$\theta$ increases, the segments $\ell_{{\theta},\downarrow}$ and $\ell_{{90^{\circ}+\theta},\uparrow}$ intersect, later they stop doing so, and later they intersect again.}
\label{fig_0_90_corners}
\end{figure}

\begin{lemma} \label{lemma:non-intersect_segm}
Let $P$ be a simple orthogonal polygon and suppose that the conditions of Lemma~\ref{lemma:empty_orthog_0_kernel} hold neither for the SE-reflex and NW-reflex vertices, nor for the SW-reflex and the NE-reflex vertices.
Let segments $\ell_{{\theta},\downarrow}$ and $\ell_{{90^{\circ}+\theta},\uparrow}$ and points $p,q$ be defined as above.
If $\ell_{{\theta},\downarrow}$ and $\ell_{{90^{\circ}+\theta},\uparrow}$ do not intersect, then $p$ and $q$ belong to
the same E-edge of $P$.
\end{lemma}

\begin{proof}
The tangency of the segments $\ell_{{\theta},\downarrow}$ and $\ell_{{90^{\circ}+\theta},\uparrow}$ to the chains $C_{SE}$ and $C_{NE}$, respectively, implies that $p,q$ belong (in fact, in that order) to
the CCW boundary chain~$\vartheta_P(t_{SE},s_{NE})$ of $P$. Suppose, for contradiction, that the point~$p$ belongs to
a N-edge. Then, no matter whether $q$ belongs to
an E-edge or
a S-edge, the left vertex of the topmost edge of the CCW boundary chain from $p$ to $q$ is a SE-reflex vertex that is higher than $p$ and thus higher than $t_{SE}$ (see vertex~$z$ in Figure~\ref{fig_0_90_corners}, middle), in contradiction to the assumption that Lemma~\ref{lemma:empty_orthog_0_kernel}, statement~(ii), does not hold for the chain~$C_{SE}$. Thus, $p$ belongs to
an E-edge.
The exact same argument enables us to show that $q$ belongs to
an E-edge, and in fact that $p, q$ belong to
the same E-edge.
\end{proof}

Since a circular arc (the locus of the intersection points of the lines supporting the rotating segments $\ell_{{\theta},\downarrow}$ and $\ell_{{90^{\circ}+\theta},\uparrow}$) and a line segment (e.g., an E-edge) intersect in at most two points (see Figure~\ref{fig_0_90_corners}, right), the above lemma implies that we may need to consider at most $3$ angular subintervals for the at most $3$ different cases to consider for the pair of $\ell_{{\theta},\downarrow}$ and $\ell_{{90^{\circ}+\theta},\uparrow}$). As there are at most $4$ such pairs, we have:

\begin{observation}\label{obs:12_subintervals}
An angle interval satisfying Property~1 (Section~\ref{subsec:opt_area_orthog_0-kernel}) may need to be broken into at most $12$ sub-intervals.
\end{observation}

Additionally, Lemma~\ref{lemma:non-intersect_segm} readily implies that if the segments $\ell_{{\theta},\downarrow}$ and $\ell_{{90^{\circ}+\theta},\uparrow}$ do not intersect, then, in the boundary of the $\{0^\circ, 90^\circ \}$-$\Kl_{\theta}(P)$, $p$ and $q$ are connected by a part of an edge of $P$. Note that the kernel has one fewer edge if $\ell_{{\theta},\downarrow}$ and $\ell_{{90^{\circ}+\theta},\uparrow}$ intersect or if exactly one of these segments rotates around a degenerate chain, that is a point (see Figure~\ref{fig_0_90_chains}, right); similar results hold for the remaining $4$~pairs of ``consecutive" segments and more occurrences of the above cases result into a kernel of even fewer edges.
Therefore:

\begin{corollary}\label{corol:kernel_size}
For $\theta \in (0, \frac{\pi}{2})$, the rotated $\{0^\circ, 90^\circ\}$-$\Kl_{\theta}(P)$ of an orthogonal polygon~$P$ has at most $8$ edges.
\end{corollary}

\subsubsection{The existence of the $\{0^\circ,90^\circ\}$-$\Kl_{\theta}(P)$ of a simple orthogonal polygon~$P$}
\label{sec:0_90-kernel-ortho-polygons}

In this subsection, we give an algorithm to determine when the $\{0^\circ,90^\circ\}$-$\Kl_{\theta}(P)$ for a simple orthogonal polygon~$P$ is non-empty. The algorithm relies on the following lemma, which is an extension of Lemma~\ref{lemma:empty_orthog_0_kernel}.

\begin{lemma} \label{lemma:empty_orthog_0_90_kernel}
If the conditions of Lemma~\ref{lemma:empty_orthog_0_kernel}
hold for either the SE-reflex and NW-reflex vertices and the chains $C_{SE}$ and $C_{NW}$, or the SW-reflex and NE-reflex vertices and the chains $C_{SW}$ and $C_{NE}$, then the $\{0^\circ,90^\circ\}$-$\Kl_{\theta}(P)$ either is empty or degenerates to a point or a line segment.
\end{lemma}

If Lemma~\ref{lemma:empty_orthog_0_90_kernel} does not apply, then we work as in Section~\ref{subsec:simple_non-empty_0-90}, namely, we determine a sequence~$\mathcal{I}$ of angle intervals such that each interval in $\mathcal{I}$ satisfies Property~1 (Section~\ref{subsec:opt_area_orthog_0-kernel}).
Then, for each such event interval~$[\gamma,\gamma')$, we find the values of $\theta \in [\gamma,\gamma')$ such that at least one of the corners of the floating rectangle~$R_{\theta}$ lies in $P$. For a corner~$r$, this computation can be easily done by comparing the position of~$r$ with the position of the corresponding endpoint, say $s$, of the line segment that bounds the strip~$S_{\theta}(P)$ and whose supporting line defines $r$, that is by comparing the locus of the corner for $\theta \in [\gamma,\gamma')$ (which is a circular arc) with the edge on which $s$ lies; see Figure~\ref{fig_0_90_corners}, left and right.
Our algorithm is outlined in Algorithm~\ref{ortho_algo3}.

\begin{algorithm}[ht!]
\caption{\label{ortho_algo3}Computing  intervals of $\theta$ such that $\{0^\circ,90^\circ\}$-$\Kl_{\theta}(P) \ne \emptyset$ for a simple orthogonal polygon~$P$}
\ \\ \
\hglue 5pt \textbf{Input:} A simple orthogonal polygon $P$ with $n$ vertices\\
\hglue 5pt \textbf{Output:} Sequence $\mathcal{E}$ of intervals for angles $\theta$ such that
$\{0^\circ,90^\circ\}$-$\Kl_{\theta}(P)\not=\emptyset$
\begin{algorithmic}[1]
    \Statex
    \Statex \hglue -6mm {\sc{\bf STEP 1:} Case for $\theta = 0$ and apply Lemma~\ref{lemma:empty_orthog_0_90_kernel}}
    \State compute the $\{0^\circ,90^\circ\}$-$\Kl_{\theta}(P)$ for $\theta = 0$
    \If {$\{0^\circ,90^\circ\}$-$\Kl_{0}(P) \ne \emptyset$}
    \State $solution\_for\_0 \gets [0,0]$
    \EndIf
    \If {the conditions of Lemma~\ref{lemma:empty_orthog_0_90_kernel} hold}
    \State apply STEP~$2$ of Algorithm~\ref{ortho_algo1} to the $\{0^\circ\}$-$\Kl_{\theta}(P)$ or $\{90^\circ\}$ -$\Kl_{\theta}(P)$ that is degenerate, compute the values \phantom{i\quad}
    of the angle~$\theta$ (if any) for which it is non-empty, and compute the subset~$solution2$ of these values for which the \phantom{i\quad} intersection of these kernels is non-empty
    \State \textbf{output} $solution\_for\_0 \cup solution2$ \ and \textbf{stop}
    \EndIf
    \Statex
    \Statex \hglue -6mm {\sc{\bf STEP 2:} Compute event intervals}
    \State compute the sequence ${\cal I}_{0^\circ}$ of subintervals having Property~1 as in STEP~2 of Algorithm~\ref{ortho_algo2}
    \State compute the sequence ${\cal I}_{90^\circ}$ of subintervals having Property~1 as in STEP~2 of Algorithm~\ref{ortho_algo2}
    \State Refine ${\cal I}_{0^\circ}$ by using ${\cal I}_{90^\circ}$ into the sequence ${\cal I} = \{I' \cap I'' = [\gamma,\gamma')\; | \; I' \in {\cal I}_{0^{\circ}} \hbox{ and } I'' \in {\cal I}_{90^{\circ}}\}$
    \Statex
    \Statex \hglue -6mm {\sc{\bf STEP 3:}   Check corners of floating rectangle}
    \For{each angle interval $[\phi,\phi')\in {\cal I}$ }
    \State by determining the values of the angle~$\theta \in [\phi,\phi')$ for which at least one of the cases (B.1), (B.2) in Lemma~\ref{lemma:obs13} \phantom{xx} \phantom{i\quad} holds, find the values of the angle~$\theta$ for which the corner lies in $P$
    \State insert these values, if any, in an initially empty sequence~$\mathcal{E}$
    \EndFor
\State {\bf output} $\mathcal{E}$
\end{algorithmic}
\end{algorithm}

Computing the $\{0^\circ,90^\circ\}$-$\Kl_{\theta}(P)$ for $\theta = 0$ can be done in $O(n)$ and space; recall that it is defined by the  lowest N-dent, by the topmost S-dent, by the rightmost W-dent, and by the leftmost E-dent. $O(n)$ time is also need to check the conditions of Lemma~\ref{lemma:empty_orthog_0_90_kernel} and $O(n)$ time and space suffice to compute the $\{0^\circ,90^\circ\}$-$\Kl_{\theta}(P)$ if any one of these conditions holds. Thus, STEP~1 can be completed in $O(n)$ time ans space. STEP~2 of Algorithm~\ref{ortho_algo2} takes $O(n)$ time and space and hence, and so does the entire STEP~2 of Algorithm~\ref{ortho_algo3}; note that the refinement of two interval sequences of $O(n)$ size each (Lemma~\ref{lemma:subinterval_number}) can be done in $O(n)$ time and produces a sequence of $O(n)$ size.
In STEP~3, checking case~(B.1) in Observation~\ref{lemma:obs13} can be done in $O(1)$ time by locating each of the $p_N(~), p_S(~)$ against the strip~$S_{\theta+90^\circ}$ and each of the $p_E(~), p_W(~)$ against the strip~$S_{\theta}$. For case~(B.2), for each of the $4$ corners, we determine the values of $\theta$, for which the circular arc traced by the corner for $\theta \in [\delta,\delta')$ intersects any of the (at most $8$) edges of the polygon that delimit the segments bounding the strips $S_{\theta}$ and $S_{\theta+90^\circ}$ (see Figure~\ref{fig_0_90_corners}); then, by taking into account whether the corner at $\theta = \delta$ lies in $P$ or not, we can find the values of the angle~$\theta$ for which the corner lies in $P$. Then, for case~(B.2), the values of $\theta$ sought are precisely the union of the angle values computed for each corner of the rectangle~$R_{\theta}$; this takes $O(1)$ time as well. Moreover, since the sequence~${\mathcal I}$ is of $O(n)$ size and because of Observation~\ref{obs:12_subintervals}, we have:

\begin{observation}
\label{obs:final_interval_size}
The total number of subintervals in the sequence~${\mathcal E}$ is $O(n)$.
\end{observation}

Finally, since the \texttt{for}-loop in STEP~3 is repeated $O(n)$ times, STEP~3 is completed in $O(n)$ time and space. Thus:

\begin{theorem}\label{thm:0_90-kernel-exists-ortho}
For a simple orthogonal polygon~$P$ with $n$ vertices, the intervals of $\theta \in [0,\frac{\pi}{2})$ for which $\{0^\circ,90^\circ\}$-$\Kl_{\theta}(P) \ne \emptyset$ can be computed in $O(n)$ time and space.
\end{theorem}

\subsubsection{Optimizing the area and perimeter of the $\{0^\circ,90^\circ\}$-$\Kl_{\theta}(P)$ of a simple orthogonal polygon~$P$}
\label{subsec:opt_area_orthog_0_90-kernel}

Our algorithm for the problem of optimizing the area/perimeter of the $\{0^\circ,90^\circ\}$-$\Kl_{\theta}(P)$ for a simple orthogonal polygon~$P$ follows the steps of Algorithm~\ref{ortho_algo3}. It treats the case for $\theta = 0$ as a special case and computes its area or perimeter, which it uses to initialize the current maximum value. Next, it checks the conditions of Lemma~\ref{lemma:empty_orthog_0_90_kernel} and computes the values of area/perimeter in these degenerate cases (see Lemma~\ref{lemma:empty_orthog_0_90_kernel}). Subsequently, it performs STEP~2 of Algorithm~\ref{ortho_algo3} and proceeds to STEP~3, except that in each small angular interval for which at least one corner of the rectangle~$R_\theta$ lies in $P$, it works incrementally maximizing the area or perimeter as in the algorithm in Section~\ref{subsec:simple_area-perimeter_0-90}, which takes $O(1)$ time; the area (resp{.} perimeter) depends linearly on $\tan\beta$, $\cot\beta$, and $\sin\beta \,\cos\beta$ (resp{.} linearly on $(1 \pm \cos\beta)/\sin\beta$, $(1 \pm \sin\beta)/\cos\beta$, and $(\sin\beta + \cos\beta)$), see the appendix.
It is important to note that for each angle interval $[\gamma,\gamma')$ with $\gamma \ne 0$ and $\gamma' \ne \frac{\pi}{2}$ for which the $\{0^\circ,90^\circ\}$-$\Kl_{\theta}(P)$ is non-empty such that the kernel is empty for $\theta = \gamma - \varepsilon$ for a small enough $\varepsilon$, the kernel for $\theta = \gamma$ is degenerate, i.e., it is a point or a line segment, so that its area and perimeter can be computed in $O(1)$ time.

The above discussion and the fact that the algorithm for the $\{0^\circ, 90^\circ \}$-$\Kl_{\theta}(P)$ is very similar to that for the $\{0^\circ\}$-$\Kl_{\theta}(P)$ lead to the following result.

\begin{theorem}\label{thm:0-90-kernel-optareaperim-ortho}
Given a simple orthogonal polygon~$P$, computing the $\{0^{\circ},90^{\circ}\}$-$\Kl_{\theta}(P)$ as well as finding an angle~$\theta$ such that its area or perimeter is maximized or minimized can be done in $O(n)$ time and space.
\end{theorem}

\subsubsection{Generalization to $k$ orientations}

For a set $\mathcal{O}$ with $k$ orientations $\alpha_1,\dots,\alpha_k$, computing the intervals of the angle~$\theta$ such that $\{\alpha_1,\dots,\alpha_k\}$-$\Kl_{\theta}(P) \ne \emptyset$ or an angle~$\theta$ such that the area or the perimeter of this kernel is reduced to computing and maintaining the intersection of $P$ with $k$ different strips. As mentioned in Section~\ref{subsec:k-orient}, Lemma~\ref{lemma:obs13} appropriately extends and the incremental construction of the kernel involves work on $O(k)$ triangles. As a result, Theorems~\ref{thm:0_90-kernel-exists-ortho} and \ref{thm:0-90-kernel-optareaperim-ortho} extend to the following theorem that leads to the results in the third row of Table~\ref{Table:ResultsOrthogonal}.

\begin{theorem}\label{thm:ortho_k_orientations}
Given a simple orthogonal polygon~$P$ with $n$ vertices, the intervals of $\theta$ such that $\{\alpha_1,\dots,\alpha_k\}$-$\Kl_{\theta}(P) \ne \emptyset$  or an angle such that the area or perimeter of $\{\alpha_1,\dots,\alpha_k\}$-$\Kl_{\theta}(P)$ are maximum/minimum can be computed in $O(kn)$ time and space.
\end{theorem}

\section*{Acknowledgements}
Alejandra Mart\'{i}nez-Moraian was funded by the predoctoral contract PRE2018-085668 of the Spanish Ministry of
Science, Innovation, and Universities and was partially supported by project PID2019-104129GB-I00 / AEI / 10.13039/501100011033 of the Spanish Ministry of Science and Innovation. David Orden was supported by project PID2019-104129GB-I00 / AEI / 10.13039/501100011033 of the Spanish Ministry of Science and Innovation and by the European Union H2020-MSCA-RISE project 734922 - CONNECT. Carlos Seara was supported by projects MTM2015-63791-R MINECO/FEDER, project PID2019-104129GB-I00 / AEI / 10.13039/501100011033 of the Spanish Ministry of Science and Innovation, Gen.~Cat. DGR 2017SGR1640, and by the European Union H2020-MSCA-RISE project 734922 - CONNECT. Pawe\l{} \.Zyli\'{n}ski was supported by grant 2015/17/B/ST6/01887 (National Science Centre, Poland).

\newpage

\appendix
\section{Appendix}

\subsection{Trigonometric formulas for the area of the $\{0^{\circ}\}$-$\Kl_{\theta}(P)$ and the $\{0^{\circ},90^{\circ}\}$-$\Kl_{\theta}(P)$}
\label{appendix:area}

We consider angle $\beta \in [\theta_i, \theta_{i+1}] \subseteq (0,\frac{\pi}{2})$ and triangles with two edges on lines forming angles~$\beta,\theta_i$ with the positive $x$-axis. For the third edge, we distinguish two cases: it is horizontal or it is on a line forming angle~$\alpha$ with the positive $x$-axis (see
Figure~\ref{fig:area_perimeter}).

\begin{figure}[htb]
\centering
\vspace{-1mm}
\subfloat{
\includegraphics[width=0.33\textwidth]{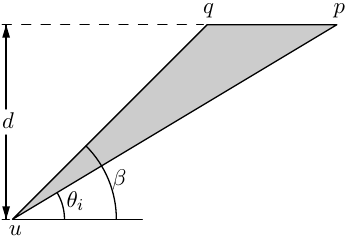}%
}
\hspace{1.5cm}
\subfloat{
\includegraphics[width=0.2\textwidth]{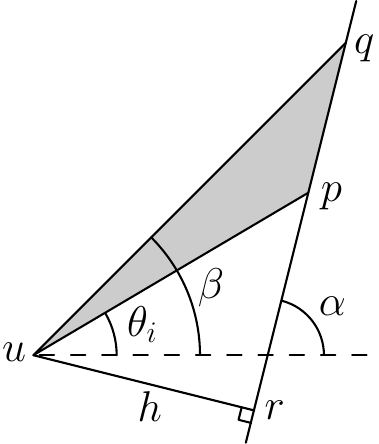}%
}
\vspace{-2mm}
\caption{For the formulas of the area and perimeter of the $\{0^{\circ}\}$-$\Kl_{\theta}(P)$.}
\label{fig:area_perimeter}
\end{figure}

From Figure~\ref{fig:area_perimeter}~(left), the area of a triangle~$T$ with edges at angles $0$ (horizontal edge), $\theta_i$, and $\beta$ ($0 < \theta_i \le \beta \le \theta_{i+1} < \frac{\pi}{2}$) is equal to
\begin{eqnarray}\label{eqn:area_h}
A_T \ =\  \frac{1}{2} \, d \  |\overline{pq}|
\ =\  \frac{1}{2} \, d \, (d\, \cot \theta_i - d\, \cot \beta)
 \ =\  \frac{1}{2} \, d^2 \, (\cot \theta_i - \frac{\cos \beta}{\sin\beta}).
\end{eqnarray}

Let us now consider a triangle~$T$ with edges at angles  $\theta_i$, $\beta$, and $\alpha$
($0 < \theta_i \le \beta < \theta_{i+1} \le \frac{\pi}{2}$) (see Figure~\ref{fig:area_perimeter}, right). Then,
$\widehat{u p r} = \alpha - \phi$ and
$\widehat{u q r} = \alpha - \theta_i$ which imply that
$|\overline{pq}| = |\overline{rq}| - |\overline{rp}| = h\, \tan (\frac{\pi}{2} - \alpha + \beta) - h\, \tan (\frac{\pi}{2} - \alpha + \theta_i)
= h\, \cot (\alpha - \beta) - h\, \cot (\alpha - \theta_i)$ where $h = |\overline{ur}|$ is the (perpendicular) distance of $u$ from the line through $p,q$. Then, the area of the triangle is equal to
\begin{eqnarray} \label{eqn:area_sl}
A_T & = & \frac{1}{2} h\, |\overline{pq}| \ =\  \frac{1}{2} h^2 \, \left( \cot (\alpha - \beta) - \cot (\alpha - \theta_i) \right) \ =\  \frac{1}{2} h^2 \, \left( \frac{1+ \cot \alpha \cot \beta}{\cot \beta - \cot \alpha} - \cot (\alpha - \theta_i) \right)\cr
\phantom{x} & = & \frac{1}{2} h^2 \, \left( \frac{\sin\beta + \cot \alpha \cos \beta}{\cos \beta - \cot \alpha \sin \beta} - \cot (\alpha - \theta_i) \right)
\end{eqnarray}

\medskip\noindent
\textbf{Expression of the area of the $\{0^\circ\}$-$\Kl_{\theta}(P)$.\ }
By using Equation~\ref{eq:area} and Equations \ref{eqn:area_h} and \ref{eqn:area_sl} and since $\theta_i$ is fixed, the area~$A(\beta)$ of the $\{0^\circ\}$-$\Kl_{\theta}(P)$
in terms of $\beta \in [\theta_i, \theta_{i+1})$ is
\[
A(\beta)= A(\theta_i) + (A_1(\beta) + A_2(\beta) - B_1(\beta) - B_2(\beta)) = \displaystyle A(\theta_i) + \sum_{i=1}^4
\left( \frac{C_i\sin\beta+D_i\cos\beta}{E_i\sin\beta+F_i\cos\beta}
+ G_i\right),
\]
where $A(\theta_i)$ is the known value of the current area, and $C_i$, $D_i$, $E_i$, $F_i$, and $G_i$ are all constants for every $i=1,\dots,4$.

Then, by setting the derivative equal to zero, we get
\[
A'(\beta) = \displaystyle\sum_{i=1}^4 \frac{C_iF_i-D_iE_i}{(E_i\sin\beta+F_i\cos\beta)^2}=0,
\]
implying that
\[
\displaystyle\sum_{i=1}^4\left[ (C_iF_i-D_iE_i) \prod_{\substack{{j=1}\\{j\neq i}}}^4 (E_j\sin\beta+F_j\cos\beta)^2 \right]=0.
\]
Expanding the product, we find three types of terms depending on $\sin^2\beta$, $\cos^2\beta$, and $\sin\beta\cos\beta$. Now using the trigonometric transformations
\[
\sin^2\beta = \frac{\tan^2\beta}{1+\tan^2\beta}, \quad
\cos^2\beta = \frac{1}{1+\tan^2\beta}, \quad \text{ and } \quad
\sin\beta\cos\beta = \frac{\tan\beta}{1+\tan^2\beta},
\]
and making the change $\tan\beta=t$ we get a rational function in $t$. Then, the derivative function for the area is now a function on the variable $t$, $A'(t)$, and it is a rational function having as numerator a polynomial in $t$ of degree $6$ and as denominator a polynomial of degree $12$. So we can compute the real solutions of a polynomial equation in $t$ of degree $6$.

\smallskip\noindent
\textbf{Orthogonal polygons:} For the case of the $\{0^{\circ}\}$-$\Kl_{\theta}(P)$ for orthogonal polygons $P$, the triangles~$A_i(\beta)$ have a horizontal or a vertical base. Since for $\alpha = \frac{\pi}{2}$, Equation~\ref{eqn:area_sl} becomes
\begin{eqnarray} \label{eqn:area:ortho}
A_T \ \ =\  \frac{1}{2} \, h^2 \, \left( \frac{\sin\beta}{\cos \beta} - \cot (0 - \theta_i) \right)
\ =\  \frac{1}{2} \, h^2 \, (\tan \beta - \tan \theta_i),
\end{eqnarray}
Equations \ref{eqn:area_h} and \ref{eqn:area:ortho} imply that in this case we have
\[
A(\beta) = A(\theta_i) + C \,\tan \beta + D \,\cot \beta + G
\]
for appropriate constants $C, D, G$.

\begin{figure}[htb]
\centering
\includegraphics[width=0.4\textwidth]{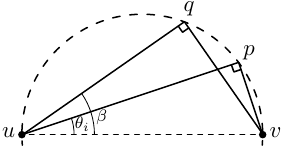}
\caption{For the formulas of the area and perimeter of the $\{0^{\circ}, 90^\circ\}$-$\Kl_{\theta}(P)$.}
\label{fig:0_90_kernel_formula}
\end{figure}

\medskip\noindent
\textbf{Expression of the area of the $\{0^\circ,90^\circ\}$-$\Kl_{\theta}(P)$.\ }
Consider the case in which the corner of the floating rectangle~$R_{\theta)}$ lies in $P$ (then it is a vertex of the kernel) and does so for all the angles $\beta \in [\theta_i,\theta_{i+1})$. Then, the corner moves along a circular arc with diameter the distance of the reflex minima/maxima that define the corner; see Figure~\ref{fig:0_90_kernel_formula}.
Then, the differential in the area is
\begin{eqnarray*}
\Delta A_T & = &  A_T({u \,v \,q}) - A_T({u \,v \,p})
\ =\  \frac{1}{2} \, |\overline{uq}| \, |\overline{vq}|
-  \frac{1}{2} \, |\overline{up}| \, |\overline{vp}|\cr
\phantom{x} & = & \frac{1}{2} \, (|\overline{uv}| \, \cos\beta) \, (|\overline{uv}| \, \sin\beta)
- \frac{1}{2} \, (|\overline{uv}| \, \cos\theta_i) \, (|\overline{uv}| \, \sin\theta_i)\cr
\phantom{x} & = & \frac{1}{2} \, |\overline{uv}|^2 \, \left( \sin\beta \, \cos\beta -  \sin\theta_i \, \cos\theta_i \right).
%%\phantom{x} & = & \frac{1}{4} \, |\overline{uv}|^2 \, (\sin 2\beta - \sin 2\theta_i).
\end{eqnarray*}
Thus, in this case, for simple polygons, the differential in the area involves at most $4$ terms, each being either
\[
\frac{C_i\sin\beta+D_i\cos\beta}{E_i\sin\beta+F_i\cos\beta}
\quad \hbox{ or } \
K_i \sin \beta \,\cos \beta.
\]

\noindent
\textbf{Orthogonal polygons:} For the case of orthogonal polygons, similarly we have at most $4$ terms, each being $C_i \tan \beta$, $D_i \cot \beta$, or $K_i \sin \beta \,\cos \beta$ and thus we have
 we have
\[
A(\beta) = A(\theta_i) + C \,\tan \beta + D \,\cot \beta + K \, \sin \beta \,\cos \beta + G
\]
for appropriate constants $C, D, K, G$.

\subsection{Trigonometric formulas for the perimeter of the $\{0^{\circ}\}$-$\Kl_{\theta}(P)$ and the $\{0^{\circ},90^{\circ}\}$-$\Kl_{\theta}(P)$}
\label{appendix:perimeter}

As in the previous section, let us first consider the case of Figure~\ref{fig:area_perimeter}, left.
We want to compute $\Delta\perim^+_T$ ($\Delta\perim^-_T$ which is the difference of the length of the edge at angle~$\beta$ minus the length of the edge at angle~$\theta_i$ plus (minus resp.) the length of the side at angle~$0$. Thus:
\begin{eqnarray}
\Delta\perim^{\pm}_T & = &  |\overline{uq}| - |\overline{up}| \pm |\overline{pq}|
\ =\  \frac{d}{\sin\beta} - \frac{d}{\sin\theta_i} \pm (d\, \cot \theta_i - d\, \cot \beta)\cr
\phantom{x} & = & d \left( \frac{1}{\sin\beta}
- \frac{1}{\sin\theta_i} \pm \frac{\cos\theta_i}{\sin\theta_i} \mp \frac{\cos\beta}{\sin\beta}
\right)
\ =\  d \left( \frac{1 \mp \cos\beta}{\sin\beta}
- \frac{1 \mp \cos\theta_i}{\sin\theta_i} \right).
\label{eqn:perim_h}
\end{eqnarray}

Next, let us consider a triangle~$T$ with edges at angles  $\theta_i$, $\beta$, and $\alpha$
($0 < \theta_i \le \beta < \theta_{i+1} \le \frac{\pi}{2}$) (see Figure~\ref{fig:area_perimeter}, right). Recall that
$\widehat{u p r} = \alpha - \phi$ and
$\widehat{u q r} = \alpha - \theta_i$ which imply that
\[
|\overline{up}| = \frac{h}{\cos(\frac{\pi}{2} - \alpha + \theta_i)} = \frac{h}{\sin(\alpha - \theta_i)},
\qquad\qquad
|\overline{uq}| = \frac{h}{\cos(\frac{\pi}{2} - \alpha + \beta} = \frac{h}{\sin(\alpha - \beta},
\]
Then, since $|\overline{pq}| = h\, \left(  \cot (\alpha - \beta) - \cot (\alpha - \theta_i) \right)$,
the differential~$\Delta\perim$ in the perimeter is equal to:
\begin{eqnarray}\label{eqn:perim_sl}
\Delta\perim^{\pm}_T & = &  |\overline{uq}| - |\overline{up}| \pm |\overline{pq}|
\ =\  \frac{h}{\sin(\alpha - \beta)} - \frac{h}{\sin(\alpha - \theta_i)} \pm  h\, \left( \cot (\alpha - \beta) - \cot (\alpha - \theta_i) \right)
\cr
\phantom{x} & = & h \left( \frac{1 \pm \cos(\alpha - \beta)}{\sin(\alpha - \beta)}
- \frac{1 \pm \cos(\alpha - \theta_i)}{\sin(\alpha - \theta_i)}
\right)\cr
\phantom{x} & = & h \left( \frac{1 \pm \cos\alpha \cos \beta \pm \sin\alpha \sin \beta}{\sin\alpha \cos\beta - \cos\alpha \sin\beta)}
- \frac{1 \pm \cos(\alpha - \theta_i)}{\sin(\alpha - \theta_i)}
\right).
\end{eqnarray}

\medskip\noindent
\textbf{Expression of the perimeter of the $\{0^\circ\}$-$\Kl_{\theta}(P)$.\ }
In Figure~\ref{Figure1-triangles}, the green (red, resp.) triangles which result in an increase (a decrease resp.) in the perimeter contribute a $\Delta\perim^{+}_T$ ($\Delta\perim^{-}_T$ resp.) term, and thus we use both the differentials $\Delta\perim^{\pm}_T$. So, from Equations \ref{eqn:perim_h} and \ref{eqn:perim_sl}, for the perimeter~$\Pi(\beta)$ of the $\{0^\circ\}$-$\Kl_{\theta}(P)$
as a function of $\beta \in [\theta_i, \theta_{i+1})$ we can write
\[
\Pi(\beta)= \displaystyle \Pi(\theta_i) + \sum_{i=1}^4
\left( \frac{C_i\sin\beta+D_i\cos\beta+H_i}{E_i\sin\beta+F_i\cos\beta}
+ G_i\right),
\]
where $\Pi(\theta_i)$ is the known value of the current perimeter, and $C_i$, $D_i$, $E_i$, $F_i$, $G_i$, and $H_i$ are all constants for every $i=1,\dots,4$.

\smallskip\noindent
\textbf{Orthogonal polygons:} For the case of the $\{0^{\circ}\}$-$\Kl_{\theta}(P)$ for an orthogonal polygons $P$, the triangles~$A_i(\beta)$ have a horizontal or a vertical base. Then from Equation~\ref{eqn:perim_sl} for $\alpha = \frac{\pi}{2}$, we have
\begin{eqnarray}\label{eqn:perim_ortho}
\Delta\perim^{\pm}_T \ =\  h \left( \frac{1 \pm \sin \beta}{\cos\beta}
- \frac{1 \pm \cos(\frac{\pi}{2}-\theta_i)}{\sin(\frac{\pi}{2}-\theta_i)}
\right) \ =\  h \left( \frac{1 \pm \sin \beta}{\cos\beta}
- \frac{1 \pm \sin\theta_i}{\cos\theta_i}
\right)
\end{eqnarray}
and Equations \ref{eqn:perim_h} and \ref{eqn:perim_ortho} imply that in this case  the perimeter~$\perim(\beta)$ of the $\{0^\circ\}$-$\Kl_{\theta}(P)$
in terms of $\beta \in [\theta_i, \theta_{i+1})$ is equal to
\[
\perim(\beta) = \perim(\theta_i)
+ C \,\frac{1 + \cos\beta}{\sin\beta}
+ D \,\frac{1 - \cos\beta}{\sin\beta}
+ E \,\frac{1 + \sin\beta}{\cos\beta}
+ F \,\frac{1 - \sin\beta}{\cos\beta}
+ G
\]
for appropriate constants $C, D, E, F, G$.

\medskip\noindent
\textbf{Expression of the perimeter of the $\{0^\circ,90^\circ\}$-$\Kl_{\theta}(P)$.\ }
In this case, we may also have corners of the kernel moving along a circular arc as shown in Figure~\ref{fig:0_90_kernel_formula}. From this figure, we observe that the differential in the perimeter is
\begin{eqnarray*}
\Delta\perim_M & = &  (|\overline{uq}| + |\overline{vq}|)
- (|\overline{up}| + |\overline{vp}|)\cr
\phantom{x} & = & (|\overline{uv}| \, \cos\beta
+ |\overline{uv}| \, \sin\beta)
- (|\overline{uv}| \, \cos\theta_i
+ |\overline{uv}| \, \sin\theta_i)\cr
\phantom{x} & = & |\overline{uv}| \, (\sin\beta + \cos\beta)
- |\overline{uv}| \, (\sin\theta_i + \cos\theta_i).
\end{eqnarray*}
Thus, for simple polygons, the differential in the perimeter involves at most $4$ terms, each being
\[
\frac{C_i\sin\beta+D_i\cos\beta+H_i}{E_i\sin\beta+F_i\cos\beta}
\quad\quad \hbox{ or } \
K_i \,(\sin\beta + \cos\beta).
\]

\noindent
\textbf{Orthogonal polygons:} For the case of orthogonal polygons, similarly we have at most $4$ terms, each being $C_i \,(1 \pm \cos\beta)/\sin\beta$, $D_i \,(1 \pm \sin\beta)/\cos\beta$, or $K_i \,(\sin\beta + \cos\beta)$ and thus
\[
\perim(\beta) = \perim(\theta_i)
+ C \,\frac{1 + \cos\beta}{\sin\beta}
+ D \,\frac{1 - \cos\beta}{\sin\beta}
+ E \,\frac{1 + \sin\beta}{\cos\beta}
+ F \,\frac{1 - \sin\beta}{\cos\beta}
+ K \,(\sin\beta + \cos\beta) + G
\]
for appropriate constants $C, D, E, F, G, K$.

\medskip
The above expressions of the perimeter~$\Pi(\beta)$ of the $\{0^\circ\}$-$\Kl_{\theta}(P)$ and the $\{0^\circ,90^\circ\}$-$\Kl_{\theta}(P)$ in terms of the angle~$\beta$ can be maximized as we showed for the area of the $\{0^\circ\}$-$\Kl_{\theta}(P)$ in Appendix~A.1 by computing the real solutions of a polynomial of constant degree.

\end{document}